\documentclass[11pt]{article}
\usepackage[margin=1in]{geometry}
\usepackage{amsmath,amssymb,amsthm,bm,mathtools}
\usepackage{hyperref}
\usepackage[capitalize,nameinlink]{cleveref}
\usepackage{algorithm}
\usepackage{algpseudocode}
\usepackage{pgfplots}
\pgfplotsset{compat=1.18} 
\usetikzlibrary{calc}

\usepackage{todonotes}

\newtheorem{theorem}{Theorem}[section]
\newtheorem{proposition}[theorem]{Proposition}

\newtheorem{corollary}[theorem]{Corollary}
\theoremstyle{definition}
\newtheorem{definition}[theorem]{Definition}
\newtheorem{remark}[theorem]{Remark}
\newtheorem{assumption}{Assumption}

\newcommand{\R}{\mathbb{R}}
\newcommand{\T}{\mathbb{T}}
\newcommand{\dd}{\,\mathrm{d}}

\newcommand{\cM}{\mathcal{M}}

\title{Additive Noise, Shift Recovery, and Signed Signals in the Cumulative Distribution Transform \thanks{This work is partially supported by the Office of Naval Research (ONR) under Award NO: N00014-24-1-2147, 
the Air Force Office of Scientific Research (AFOSR) under Award NO: FA9550-25-1-0231, and NSF grantDMS-2408877.
}}
\author{Harbir Antil$^{1}$, Ratna Khatri$^2$, and Aryan Saxena$^1$ \\
$^{1}$Center for Mathematics and Artificial Intelligence and \\ Department of Mathematical Sciences, George Mason University, \\ Fairfax, Virginia 22030. \\
$^{2}$U.S. Naval Research Laboratory, Washington D.C. 
}
\date{\today}

\begin{document}
\maketitle

\begin{abstract}
The cumulative distribution transform (CDT) is a quantile-based transport representation that exactly linearizes one-dimensional translations of positive densities. We study how this structure behaves under additive perturbations and how it can be exploited for shift recovery. Under a local nondegeneracy condition, we derive a first-order expansion showing that additive noise in physical space induces a nonlocal perturbation in CDT space through the primitive of the noise, weighted by the reciprocal density. This yields an explicit description of transform-domain sensitivity and shows, in particular, that perturbations are amplified in low-density regions. When the physical-space perturbation is modeled as a centered Gaussian random field, the induced first-order CDT perturbation is again Gaussian, with an explicit covariance kernel.

We then use this structure to study recovery in CDT coordinates. In the known-template setting, the transport shift is obtained by projection onto the constant mode, giving an explicit estimator together with exactness in the noiseless case and a stability bound under perturbations. In the unknown-template setting, multiple observations permit joint recovery of the shifts and a common template up to the natural constant-mode gauge, leading to a simple de-shift--and--average procedure. We also consider a signed-signal analogue based on the signed cumulative distribution transform (SCDT), where shifts are estimated numerically by feature matching and unknown templates are recovered by alternating alignment and averaging. Numerical experiments validate the perturbation analysis and illustrate effective recovery for both density-valued and signed signals.
\end{abstract}


\section{Introduction}

Transport-based signal representations have become increasingly useful for analyzing variability that is dominated by displacement rather than amplitude. In one spatial dimension, the cumulative distribution transform (CDT) provides a particularly transparent example: it represents a positive normalized density by its quantile map relative to a fixed reference density, so that translations become additive constants in transform space \cite{ParkKolouriKunduRohde2018}. In this sense, the CDT may be viewed as a one-dimensional quantile-based transport representation \cite{AShapiro_DDentcheva_ARuszczynski_2014a,JORoyset_2025a} closely connected to optimal transport ideas \cite{Villani2009,PeyreCuturi2019}. This exact linearization property has made the CDT useful in transform-domain signal analysis, classification, and related transport-based modeling tasks \cite{ParkKolouriKunduRohde2018,KolouriParkRohde2016}. Related CDT-based estimation ideas have also appeared in \cite{RubaiyatHallamNicholsHutchinsonLiRohde2020}, where the transform is used for parametric signal estimation problems including time delay, linear and quadratic dispersion, and more general polynomial time-warps. See also \cite{SThareja_GRohde_RDMartin_IMedri_AAldroubi_2022a} for the signed case.

In many applications, however, transport variability is not observed in isolation. Measurements are corrupted by additive perturbations, templates may be unknown, and in some settings the signals of interest are signed rather than strictly nonnegative densities. These issues raise a natural question: how stable is the CDT under additive noise, and to what extent can its transport geometry still be exploited for recovery and alignment? The aim of this paper is to address that question in a one-dimensional setting.

Our starting point is the exact transport identity for translations in CDT space. We then study additive perturbations of a density \(u\) of the form \(u+\delta\eta\), where \(\eta\) has zero total mass. Under a local nondegeneracy condition at the relevant quantile location, we derive a first-order expansion for the perturbed CDT. This expansion shows that additive noise is transformed through its primitive and weighted by the reciprocal of the density. As a result, the CDT acts nonlocally on additive perturbations and exhibits amplified sensitivity in low-density regions. When the perturbation field is Gaussian, the induced first-order CDT perturbation is again Gaussian, and its covariance can be written explicitly in terms of the covariance kernel of the original noise. This gives a precise probabilistic description of transform-space noise at leading order. While \cite{RubaiyatHallamNicholsHutchinsonLiRohde2020} also contains a statistical discussion of noise in CDT/Wasserstein coordinates, the setting and conclusions are different: that work is tied to a normalized noisy-signal model and an associated correction of the cumulative distribution function, whereas here we derive a local first-order perturbation formula for the CDT itself under explicit regularity assumptions.

The same first-order structure also leads naturally to recovery procedures. In many applications, the dominant variability is not only noisy but also transport-driven, so one would like to estimate and remove shifts before performing averaging or inference. In the known-template setting, the transport contribution appears as a constant mode in CDT space, while the perturbation occupies the remaining directions. This makes it possible to estimate the shift by projection onto the constant mode and to isolate an observable de-shifted residual. In the unknown-template setting, a single observation is insufficient for full recovery, but multiple translated observations sharing a common template permit joint estimation of the shifts and the template up to the natural constant-mode gauge. This yields a simple de-shift--and--average algorithm in CDT space. Thus, while prior CDT-based estimation work has focused on broader parametric deformation models \cite{RubaiyatHallamNicholsHutchinsonLiRohde2020}, our emphasis here is on additive perturbations, shift recovery, and template estimation in the presence of noise.

To extend the recovery viewpoint beyond densities, we also consider the signed cumulative distribution transform (SCDT), which applies transport-based coordinates to the positive and negative parts of a signed signal \cite{AldroubiDiazMartinMedriRohdeThareja2022}. In contrast to the density-valued CDT setting, no closed-form constant-mode shift estimator is available there. We therefore formulate numerical SCDT analogues of the known-template and unknown-template recovery procedures based on shift-grid matching in feature space together with alignment and averaging in signal space. Recent work has also shown the usefulness of SCDT representations in learning and classification settings \cite{RubaiyatLiYinShifatERabbiZhuangRohde2024}.

The paper is organized as follows. Section~\ref{sec:preliminaries} reviews CDFs, quantiles, and the CDT. CDT linearization of translations and geometric properties of the CDT are also stated. Section~\ref{sec:noise} derives the first-order perturbation formula and characterizes the Gaussian structure of linearized CDT noise. Section~\ref{sec:recovery} develops CDT-based shift estimation and joint template recovery. Section~\ref{sec:scdt} introduces the corresponding SCDT recovery procedures for signed signals. Section~\ref{sec:numerics} presents numerical experiments validating the perturbation analysis and illustrating recovery for both CDT and SCDT.

\section{Preliminaries: CDFs, Quantiles, and the CDT}
\label{sec:preliminaries}

\begin{assumption}[Strictly positive admissible densities]\label{ass:density}
Throughout the paper, a density \(w:\R\to\R\) is assumed to satisfy
\[
w\in C(\R),\qquad w(x)>0 \ \text{for all }x\in\R,\qquad \int_{\R} w(x)\,\dd x = 1.
\]
\end{assumption}

\begin{remark}
A probability density need only be nonnegative, integrable, and have unit mass. The stronger assumption in \cref{ass:density} is imposed here to ensure that the associated CDF is strictly increasing, so that the quantile map is well defined as an ordinary inverse, and to exclude degeneracies in the perturbation arguments below.
\end{remark}

\begin{definition}[CDF and quantile]\label{def:cdf_quantile}
Given \(w\) satisfying \cref{ass:density}, define its cumulative distribution function by
\[
W(x):=\int_{-\infty}^x w(t)\,\dd t,\qquad x\in\R.
\]
Then \(W:\R\to(0,1)\) is continuous and strictly increasing. Its quantile function is defined by
\[
Q_w(p):=W^{-1}(p),\qquad p\in(0,1).
\]
\end{definition}

\begin{assumption}[Reference density and weighted space]\label{ass:ref}
Fix a reference density \(r\) satisfying \cref{ass:density}, with CDF
\[
R(\alpha):=\int_{-\infty}^{\alpha} r(s)\,\dd s,\qquad \alpha\in\R.
\]
We also introduce the weighted Hilbert space
\[
L^2_r
:=
\left\{
f:\R\to\R:
\|f\|_{L^2_r}^2
:=
\int_{\R}|f(\alpha)|^2\,r(\alpha)\,\dd\alpha<\infty
\right\}.
\]
For \(f\in L^2_r\), define its \(r\)-weighted mean by
\[
\bar f:=\int_{\R} f(\alpha)\,r(\alpha)\,\dd\alpha.
\]
\end{assumption}

\begin{definition}[CDT / transport map representation]\label{def:cdt}
Let \(w\) satisfy \cref{ass:density}, with CDF \(W\), and let \(R\) be the reference CDF from \cref{ass:ref}. The cumulative distribution transform (CDT) of \(w\) relative to \(R\) is the function
\[
\widehat w(\alpha):=Q_w(R(\alpha))=W^{-1}(R(\alpha)),\qquad \alpha\in\R.
\]
Equivalently, \(\widehat w(\alpha)\) is the unique point satisfying
\[
W(\widehat w(\alpha))=R(\alpha),\qquad \alpha\in\R.
\]
\end{definition}

\subsection{CDT Linearizes Translations Exactly}

The following result is well-known but we recall it for completeness.
\begin{theorem}[Translation linearization]\label{thm:linearize}
Let \(w\) satisfy \cref{ass:density}. For any shift \(s\in\R\), define
\[
w_s(x):=w(x-s).
\]
Then the CDT satisfies
\[
\widehat{w_s}(\alpha)=\widehat w(\alpha)+s
\qquad
\text{for all }\alpha\in\R.
\]
\end{theorem}

\begin{proof}
Let \(W\) be the CDF of \(w\). Then the CDF of \(w_s\) is
\[
W_s(x)=\int_{-\infty}^x w(t-s)\,\dd t = W(x-s).
\]
Fix \(\alpha\in\R\) and set \(p:=R(\alpha)\in(0,1)\). By definition of the quantile function, \(Q_{w_s}(p)\) is the unique point satisfying
\[
W_s(Q_{w_s}(p))=p.
\]
Using \(W_s(x)=W(x-s)\), this becomes
\[
W(Q_{w_s}(p)-s)=p.
\]
Since \(Q_w(p)\) is the unique point such that \(W(Q_w(p))=p\), it follows that
\[
Q_{w_s}(p)=Q_w(p)+s.
\]
Evaluating at \(p=R(\alpha)\) gives
\[
\widehat{w_s}(\alpha)
=
Q_{w_s}(R(\alpha))
=
Q_w(R(\alpha))+s
=
\widehat w(\alpha)+s,
\]
which proves the claim.
\end{proof}

\subsection{Geometric Structure and ROM Implications}
\label{sec:geometry}

Let \(g\) satisfy \cref{ass:density} and consider the translation family
\[
\cM_g:=\{\,g(\cdot-s):\ s\in[0,S]\,\}\subset L^2(\R)
\qquad
(\text{or }L^2(\T)\text{ in the periodic case}).
\]
By \cref{thm:linearize},
\[
\widehat{\cM_g}
:=
\{\widehat{g(\cdot-s)}:\ s\in[0,S]\}
=
\{\widehat g+s:\ s\in[0,S]\}
\subset \widehat g+\mathrm{span}\{1\}.
\]
Thus the translated family becomes an affine segment in CDT space, with direction given by the constant mode \(1\).

\begin{remark}[Geometric implication]
In physical space, translated families are typically poorly approximated by low-dimensional linear spaces, which leads to slow singular value decay and limits the effectiveness of standard linear reduced-order models. In CDT space, the same family becomes affine and is therefore governed by a very low-dimensional geometric structure. This simple observation is the basic reason CDT coordinates are well suited to transport-dominated variability, and it underlies both the alignment procedures and the numerical behavior observed later in the paper.
\end{remark}

\section{First-Order Noise Effects in CDT Space}
\label{sec:noise}

\subsection{First-Order CDT Perturbation}

Fix \(\alpha\in\R\) and consider a perturbation of the form
\[
u_\delta = u + \delta\eta,
\qquad
\int_{\R}\eta(x)\,\dd x = 0,
\qquad
E(x):=\int_{-\infty}^x \eta(s)\,\dd s.
\]
We assume that for all sufficiently small \(\delta\), the function \(u_\delta\) remains an admissible density.

\begin{theorem}[First-order CDT perturbation]\label{thm:cdt_noise}
Let \(u\) satisfy \cref{ass:density}, and fix \(\alpha\in\R\). Assume that there exists \(c_\alpha>0\) such that
\[
u(\widehat u(\alpha))\ge c_\alpha.
\]
Then, as \(\delta\to 0\),
\begin{equation}\label{eq:cdt_noise_expansion}
\widehat{u+\delta\eta}(\alpha)
=
\widehat u(\alpha)
-\delta\,\frac{E(\widehat u(\alpha))}{u(\widehat u(\alpha))}
+o(\delta).
\end{equation}
\end{theorem}

\begin{proof}
Let \(U\) be the CDF of \(u\) and let \(U_\delta\) be the CDF of \(u_\delta=u+\delta\eta\). Since
$
\int_{\R}\eta(x)\,\dd x = 0,
$
we have
\[
U_\delta(x)=U(x)+\delta E(x).
\]
Fix \(\alpha\in\R\) and set
$
p:=R(\alpha)\in(0,1).
$
By definition, \(\widehat u(\alpha)\) is the unique solution of
\[
U(\widehat u(\alpha))=p,
\]
and, for sufficiently small \(|\delta|\), \(\widehat{u+\delta\eta}(\alpha)\) is the unique solution of
\[
U_\delta(\widehat{u+\delta\eta}(\alpha))=p.
\]
Let
\[
x_\delta:=\widehat{u+\delta\eta}(\alpha),
\qquad
x_0:=\widehat u(\alpha).
\]
Then
\[
0
=
p-U_\delta(x_\delta)
=
p-U(x_\delta)-\delta E(x_\delta).
\]
Subtracting the identity
\[
0=p-U(x_0)
\]
gives
\[
U(x_\delta)-U(x_0)+\delta E(x_\delta)=0.
\]
Since \(U'(x)=u(x)\) and \(u(x_0)\ge c_\alpha>0\), a first-order Taylor expansion yields
\[
U(x_\delta)-U(x_0)
=
u(x_0)(x_\delta-x_0)+o(|x_\delta-x_0|).
\]
Also, since \(E\) is continuous,
\[
E(x_\delta)=E(x_0)+o(1)
\qquad
\text{as }x_\delta\to x_0.
\]
Therefore
\[
u(x_0)(x_\delta-x_0)+\delta E(x_0)+o(|x_\delta-x_0|)+\delta o(1)=0.
\]
Writing \(h_\delta:=x_\delta-x_0\), we obtain
\[
u(x_0)h_\delta
=
-\delta E(x_0)-o(|h_\delta|)-\delta o(1).
\]
Since \(u(x_0)\ge c_\alpha>0\), it follows that \(h_\delta=O(\delta)\). Hence \(o(|h_\delta|)=o(\delta)\), and so
\[
h_\delta
=
-\delta \frac{E(x_0)}{u(x_0)}+o(\delta),
\]
which is exactly \eqref{eq:cdt_noise_expansion}.
\end{proof}

\begin{remark}[Meaning of the nondegeneracy assumption]
The condition
\[
u(\widehat u(\alpha))\ge c_\alpha>0
\]
is a local nondegeneracy condition at the quantile location \(\widehat u(\alpha)\). It ensures that the CDF is not too flat there, so that small perturbations of cumulative mass produce controlled perturbations of the quantile. For densities such as Gaussians, this condition holds automatically for every fixed \(\alpha\), although the lower bound need not be uniform in \(\alpha\); in particular, it may become small in the far tails.
\end{remark}

\begin{remark}[Transport and perturbation are separated in CDT space]
If \(u=u_0(\cdot-s)\) is a translated template, then \cref{thm:linearize,thm:cdt_noise} yield the additive first-order decomposition
\[
\widehat u(\alpha)
=
\widehat u_0(\alpha)+s
-\delta\,\frac{E(\widehat u_0(\alpha))}{u_0(\widehat u_0(\alpha))}
+o(\delta).
\]
Thus the transport contribution appears as the scalar shift \(s\), while the perturbation enters through the primitive \(E\), rescaled by the reciprocal density. In particular, low-density regions amplify perturbations.
\end{remark}

\subsection{Gaussian Structure of Linearized CDT Noise}

We now characterize the distributional structure of the first-order CDT perturbation. We begin by defining the linear operator that governs how additive perturbations in physical space map to displacements in the CDT domain.

\begin{definition}[Linearized CDT operator]
\label{def:cdt_operator}
For a given density $u$, we define the linearized CDT operator $\mathcal{L}_u$ acting on a perturbation $\eta$ as
\[
(\mathcal{L}_u \eta)(\alpha) := -\frac{1}{u(\widehat{u}(\alpha))} \int_{-\infty}^{\widehat{u}(\alpha)} \eta(s) \, \dd s.
\]
\end{definition}

By \cref{thm:cdt_noise}, the CDT of the noisy signal satisfies
\[
\widehat{u+\delta\eta}(\alpha) = \widehat{u}(\alpha) + \delta (\mathcal{L}_u \eta)(\alpha) + o(\delta).
\]
This operator reveals that the CDT does not act locally: it first integrates the noise and then rescales it by the reciprocal of the template density.

\subsubsection*{Gaussian Perturbations}

\begin{assumption}[Centered Gaussian noise]
\label{ass:gaussian_noise}
The perturbation $\eta$ is a centered Gaussian random field with covariance kernel $C_\eta(s,t) := \operatorname{Cov}(\eta(s), \eta(t))$. We assume $\eta \in L^1(\mathbb{R})$ and $\int_{\mathbb{R}} \eta = 0$ almost surely, with $C_\eta$ sufficiently integrable over the domain of interest.
\end{assumption}

\begin{theorem}[Induced CDT noise]
\label{thm:cdt_gaussian}
Under \cref{ass:gaussian_noise}, the linearized CDT noise $\xi := \mathcal{L}_u \eta$ is a centered Gaussian random field. Its covariance kernel is given by
\[
\operatorname{Cov}(\xi(\alpha), \xi(\beta)) = \frac{1}{u(\widehat{u}(\alpha))u(\widehat{u}(\beta))} \int_{-\infty}^{\widehat{u}(\alpha)} \int_{-\infty}^{\widehat{u}(\beta)} C_\eta(s, t) \, \dd s \, \dd t.
\]
\end{theorem}

\begin{proof}
Since $\mathcal{L}_u$ is a linear operator, the field $\xi$ is a linear transformation of the Gaussian field $\eta$. Its finite-dimensional marginals are therefore multivariate Gaussian. The mean is zero by the linearity of expectation: $\mathbb{E}[\mathcal{L}_u \eta] = \mathcal{L}_u \mathbb{E}[\eta] = 0$. The covariance follows by substituting the definition of $\xi$ and applying the linearity of the covariance operator to the integral:
\[
\operatorname{Cov}(\xi(\alpha), \xi(\beta)) = \mathbb{E} \left[ \frac{E(\widehat{u}(\alpha))}{u(\widehat{u}(\alpha))} \frac{E(\widehat{u}(\beta))}{u(\widehat{u}(\beta))} \right],
\]
where $E(x) = \int_{-\infty}^x \eta(s) \, \dd s$. Expanding the product of integrals yields the result.
\end{proof}

\begin{corollary}[First-order CDT covariance]
\label{cor:cdt_stats}
Under the first-order approximation $\widehat{u+\delta\eta} = \widehat{u} + \delta\xi + o(\delta)$, the covariance of the observed CDT satisfies
\[
\operatorname{Cov}\bigl(\widehat{u+\delta\eta}(\alpha),\widehat{u+\delta\eta}(\beta)\bigr)
=
\delta^2 \operatorname{Cov}(\xi(\alpha),\xi(\beta)) + o(\delta^2).
\]
In particular, the variance of the CDT at coordinate $\alpha$ scales with the square of the perturbation magnitude $\delta$ and the inverse square of the local density.
\end{corollary}

The structure of $\mathcal{L}_u$ highlights how physical noise is ``recolored" in the CDT domain:
\begin{itemize}
    \item \textbf{Spatial Correlation:} Because of the integration step, even spatially uncorrelated noise (e.g., white noise) becomes correlated in CDT space.
    \item \textbf{Heteroscedasticity:} The $1/u(\widehat{u}(\alpha))$ term acts as a signal-dependent gain, amplifying noise in regions where the template density is low.
    \item \textbf{Nonlocality:} The covariance at $(\alpha, \beta)$ depends on the total mass of the noise covariance kernel up to the corresponding points in physical space.
\end{itemize}

\subsubsection*{Connection to the additive CDT model}

If \(u=u_0(\cdot-s)\) is a translated template, then \cref{thm:linearize,thm:cdt_noise} yield the first-order decomposition
\[
\widehat u_{\mathrm{obs}}(t,\alpha)
=
\widehat u_0(\alpha)+s(t)+\delta\xi(t,\alpha),
\]
with first-order perturbation field
\[
\xi(t,\alpha)
=
-\frac{1}{u_0(\widehat u_0(\alpha))}
\int_{-\infty}^{\widehat u_0(\alpha)} \eta(t,s)\,\dd s.
\]
Thus, in CDT space, the leading transport contribution appears as a scalar shift, while the noise enters through a structured linear integral operator applied to the physical-space perturbation.

\section{Shift Estimation and Residual Separation in CDT Space}
\label{sec:recovery}

\subsection{Known-Template Recovery}

We first consider the translated-template setting in which the observed CDT signal satisfies
\begin{equation}\label{eq:obs_model_cdt}
\widehat u_{\mathrm{obs}}(t,\alpha)
=
\widehat u_0(\alpha)+s(t)+\delta\xi(t,\alpha),
\end{equation}
where:
\begin{itemize}
    \item \(\widehat u_{\mathrm{obs}}(t,\cdot)\) is observed,
    \item the template \(\widehat u_0\) is known,
    \item the shift \(s(t)\) is unknown,
    \item the perturbation field \(\xi(t,\alpha)\) is unknown,
    \item and the noise amplitude \(\delta\) is unknown.
\end{itemize}
Thus the goal is not to identify \(\delta\) and \(\xi\) separately from a single observation, but rather to recover the transport contribution and the observable de-shifted residual. The geometric structure of this decomposition is illustrated in \cref{fig:projection_geom}.

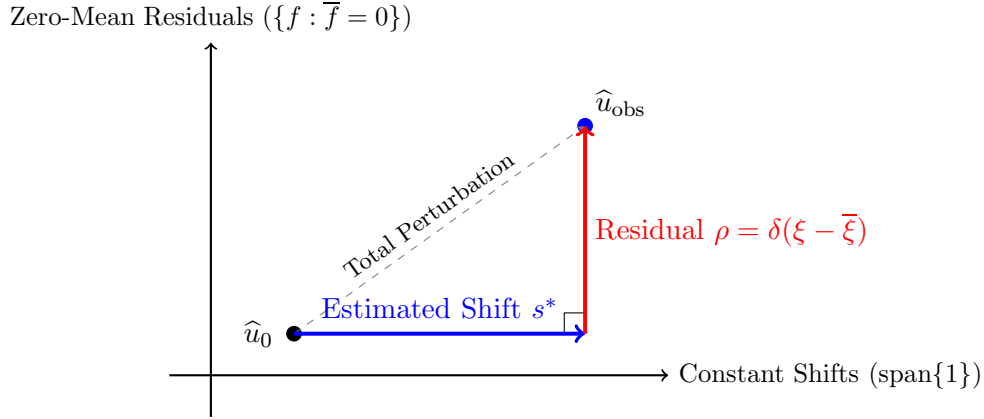
\begin{figure}[ht]
\centering
\begin{tikzpicture}[scale=1.1]
    \draw[->, thick] (-0.5,0) -- (5.5,0) node[right, font=\small] {Constant Shifts ($\text{span}\{1\}$)};
    \draw[->, thick] (0,-0.5) -- (0,4.0) node[above, font=\small] {Zero-Mean Residuals ($\{f: \overline{f}=0\}$)};
    
    \coordinate (H) at (1,0.5);    
    \coordinate (Y) at (4.5,3.0);  
    \coordinate (Proj) at (4.5,0.5); 
    
    \filldraw[black] (H) circle (2.5pt);
    \node[anchor=east, xshift=-4pt] at (H) {\(\widehat{u}_0\)};
    
    \filldraw[blue] (Y) circle (2.5pt);
    \node[anchor=south west] at (Y) {\(\widehat{u}_{\mathrm{obs}}\)};
    
    \draw[->, >=stealth, dashed, gray, thin] (H) -- (Y) node[midway, sloped, above, black, font=\footnotesize] {Total Perturbation};
    
    \draw (4.25,0.5) -- (4.25,0.75) -- (4.5,0.75);
    
    \draw[line width=1.5pt, blue, ->] (1,0.5) -- (Proj);
    \node[anchor=south, blue, yshift=2pt] at (2.75, 0.5) {Estimated Shift \(s^*\)};
    
    \draw[line width=1.5pt, red, ->] (Proj) -- (Y);
    \node[anchor=west, red] at (4.5,1.75) {Residual \(\rho = \delta(\xi - \overline{\xi})\)};

\end{tikzpicture}
\caption{Geometric decomposition of the CDT perturbation in \(L^2_r\). The transport component lies in the constant mode \(\mathrm{span}\{1\}\), while the centered residual lies in its orthogonal complement. Projecting the observed change onto the constant mode therefore isolates the shift estimate \(s^\ast\).}
\label{fig:projection_geom}
\end{figure}

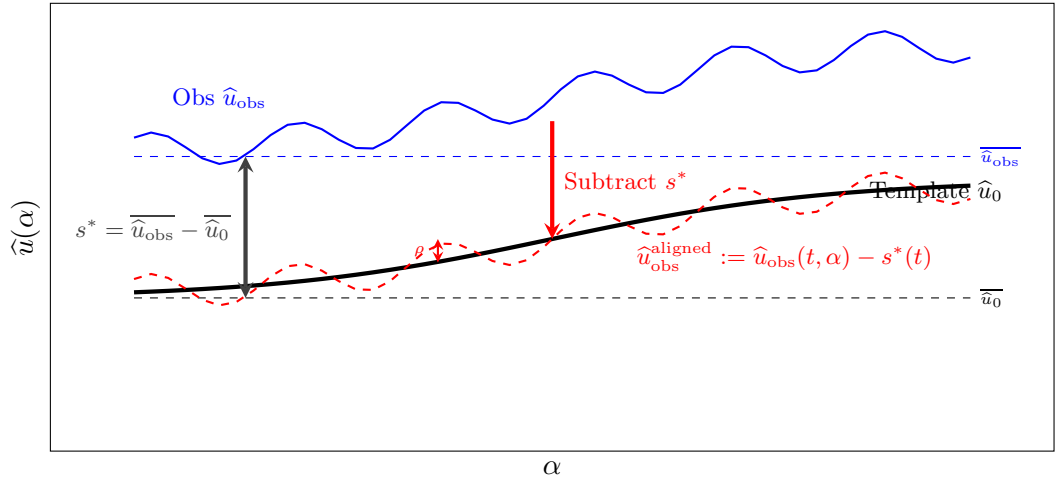
\begin{figure}[h!]
\centering
\begin{tikzpicture}
\begin{axis}[
    width=0.9\textwidth, height=7.5cm,
    title={De-shifting with Known Template (\cref{alg:cdt_deshift})},
    xlabel={$\alpha$}, ylabel={$\widehat{u}(\alpha)$},
    ticks=none, domain=-3:3, samples=50,
    ymin=-1.0, ymax=2.8,
    legend style={at={(0.5,-0.2)}, anchor=north, legend columns=-1},
    clip=false
]

    \addplot[black, ultra thick] {1/(1+exp(-x)) + 0.3}; 
    \node at (axis cs: 2.2, 1.2) [black, anchor=west, font=\footnotesize] {Template $\widehat{u}_0$};
    
    \draw[black, dashed, thin] (axis cs: -3, 0.3) -- (axis cs: 3, 0.3);
    \node at (axis cs: 3.0, 0.3) [black, anchor=west, font=\tiny] {$\overline{\widehat{u}_0}$};

    \addplot[blue, opacity=0.5, thick] {1/(1+exp(-x)) + 1.5 + 0.15*sin(deg(6*x))}; 
    \node at (axis cs: -2.8, 2.0) [blue, anchor=west, font=\footnotesize] {Obs $\widehat{u}_{\mathrm{obs}}$};
    
    \draw[blue, dashed, thin] (axis cs: -3, 1.5) -- (axis cs: 3, 1.5);
    \node at (axis cs: 3.0, 1.5) [blue, anchor=west, font=\tiny] {$\overline{\widehat{u}_{\mathrm{obs}}}$};

    \coordinate (MeanBlk) at (axis cs: -2.2, 0.3);
    \coordinate (MeanBlu) at (axis cs: -2.2, 1.5);
    \draw[<->, >=stealth, darkgray, ultra thick] (MeanBlk) -- (MeanBlu);
    \node[darkgray, anchor=east, font=\footnotesize, xshift=-2pt] at (axis cs: -2.2, 0.9) {$s^* = \overline{\widehat{u}_{\mathrm{obs}}} - \overline{\widehat{u}_0}$};

    \addplot[red, thick, dashed] {1/(1+exp(-x)) + 0.3 + 0.15*sin(deg(6*x))};
    \node at (axis cs: 2.8, 0.65) [red, anchor=east, font=\footnotesize] {$\widehat{u}_{\mathrm{obs}}^{\mathrm{aligned}} := \widehat u_{\mathrm{obs}}(t,\alpha)-s^\ast(t)$};

    \draw[->, >=stealth, ultra thick, red] (axis cs: 0, 1.8) -- (axis cs: 0, 0.8) 
        node[midway, right, font=\footnotesize, text width=2cm] {Subtract $s^*$};

    \coordinate (RedPt) at (axis cs: -0.82, 0.8);
    \coordinate (BlkPt) at (axis cs: -0.82, 0.6);
    \draw[<->, >=stealth, red, thick] (RedPt) -- (BlkPt) 
        node[midway, left, font=\tiny, xshift=-1pt] {$\rho$};

\end{axis}
\end{tikzpicture}
\caption{Visualizing the de-shifting process in \cref{alg:cdt_deshift}. The dashed lines indicate the \(r\)-weighted means of the template and the observation. Their difference gives the shift estimate \(s^\ast\). Subtracting \(s^\ast\) from the observation yields the aligned signal, while the remaining shape discrepancy is the residual \(\rho\).}

\label{fig:single_pipeline_full_v2}
\end{figure}

The key point is that in \eqref{eq:obs_model_cdt} the shift \(s(t)\) is constant in \(\alpha\), so it lies in the one-dimensional subspace \(\mathrm{span}\{1\}\subset L^2_r\). As shown in \cref{fig:projection_geom}, the orthogonality of the constant mode and the centered residuals suggests estimating the shift by projection onto the constant mode. The algorithm is visualized in Figure~\ref{fig:single_pipeline_full_v2} and described in detail in \cref{alg:cdt_deshift}. 

For fixed \(t\), define
\[
y(\alpha):=\widehat u_{\mathrm{obs}}(t,\alpha),
\qquad
h(\alpha):=\widehat u_0(\alpha).
\]
Consider the least-squares problem
\begin{equation}\label{eq:shift_ls_single}
\min_{s\in\R}\ \frac12\|y-(h+s)\|_{L^2_r}^2.
\end{equation}
This seeks the best approximation of the observed CDT signal by the affine line \(h+\mathrm{span}\{1\}\subset L^2_r\).

\begin{proposition}[Explicit shift estimator, exactness, and stability]\label{prop:shift_est}
For fixed \(t\), the unique minimizer of \eqref{eq:shift_ls_single} is
\[
s^\ast(t)=\overline{\widehat u_{\mathrm{obs}}(t,\cdot)}-\overline{\widehat u_0}.
\]
Equivalently, \(s^\ast(t)\) is the coefficient obtained by projecting
\[
\widehat u_{\mathrm{obs}}(t,\cdot)-\widehat u_0
\]
onto \(\mathrm{span}\{1\}\subset L^2_r\).
Moreover:
\begin{enumerate}
\item If
\[
\widehat u_{\mathrm{obs}}(t,\alpha)=\widehat u_0(\alpha)+s(t),
\]
then
\[
s^\ast(t)=s(t).
\]
Thus the estimator is exact in the noiseless translated-template case.

\item If \eqref{eq:obs_model_cdt} holds with \(\xi(t,\cdot)\in L^2_r\), then
\[
s^\ast(t)=s(t)+\delta\,\overline{\xi(t,\cdot)},
\]
and therefore
\[
|s^\ast(t)-s(t)|\le |\delta|\,\|\xi(t,\cdot)\|_{L^2_r}.
\]
\end{enumerate}
\end{proposition}

\begin{proof}
Fix \(t\) and define
\[
\Phi(s):=\frac12\|y-(h+s)\|_{L^2_r}^2.
\]
This is a strictly convex quadratic in \(s\), so the unique minimizer satisfies \(\Phi'(s)=0\). Differentiating gives
\[
\Phi'(s)
=
-\int_{\R}\bigl(y(\alpha)-h(\alpha)-s\bigr)\,r(\alpha)\,\dd\alpha
=-(\bar y-\bar h-s).
\]
Hence
\[
s^\ast=\bar y-\bar h.
\]
If \(y=h+s(t)\), then \(s^\ast=s(t)\).
If instead \(y=h+s(t)+\delta\xi\), then
\[
s^\ast=\bar y-\bar h=s(t)+\delta\,\overline{\xi(t,\cdot)}.
\]
Since \(\int_\R r(\alpha)\,\dd\alpha=1\), Cauchy--Schwarz in \(L^2_r\) yields
\[
\bigl|\overline{\xi(t,\cdot)}\bigr|
=
\left|\int_{\R}\xi(t,\alpha)\,r(\alpha)\,\dd\alpha\right|
\le
\|\xi(t,\cdot)\|_{L^2_r}.
\]
Therefore
\[
|s^\ast(t)-s(t)|\le |\delta|\,\|\xi(t,\cdot)\|_{L^2_r}.
\qedhere
\]
\end{proof}

The next result identifies the observable residual after de-shifting.

\begin{proposition}[Residual after de-shifting]\label{prop:deshift_residual}
Under the assumptions of \cref{prop:shift_est}, define
\[
\rho(t,\alpha)
:=
\widehat u_{\mathrm{obs}}(t,\alpha)-\widehat u_0(\alpha)-s^\ast(t).
\]
Then
\[
\rho(t,\alpha)
=
\delta\Bigl(\xi(t,\alpha)-\overline{\xi(t,\cdot)}\Bigr).
\]
In particular,
\[
\overline{\rho(t,\cdot)}=0,
\]
so the residual lies in the orthogonal complement of \(\mathrm{span}\{1\}\) in \(L^2_r\).
\end{proposition}

\begin{proof}
By \cref{prop:shift_est},
\[
s^\ast(t)=s(t)+\delta\,\overline{\xi(t,\cdot)}.
\]
Hence
\[
\rho(t,\alpha)
=
\bigl(\widehat u_0(\alpha)+s(t)+\delta\xi(t,\alpha)\bigr)
-\widehat u_0(\alpha)
-\bigl(s(t)+\delta\,\overline{\xi(t,\cdot)}\bigr),
\]
which simplifies to
\[
\rho(t,\alpha)=\delta\Bigl(\xi(t,\alpha)-\overline{\xi(t,\cdot)}\Bigr).
\]
Taking the \(r\)-weighted mean gives
\[
\overline{\rho(t,\cdot)}
=
\delta\Bigl(\overline{\xi(t,\cdot)}-\overline{\xi(t,\cdot)}\Bigr)=0.
\]
\end{proof}

\begin{remark}[What is and is not recovered]
Given a known template \(\widehat u_0\), the observable CDT signal determines:
\begin{itemize}
\item the shift estimate \(s^\ast(t)\),
\item the de-shifted residual \(\rho(t,\alpha)\).
\end{itemize}
It does not determine \(\delta\) and \(\xi\) separately from a single observation, and it does not determine an unknown template \(\widehat u_0\) without further information or multiple observations.
\end{remark}

\begin{algorithm}[h!]
\caption{CDT-Based De-shifting with Known Template}
\label{alg:cdt_deshift}
\begin{algorithmic}[1]
\Require Observed signal(s) \(u_{\mathrm{obs}}(x,t)\), known template \(u_0(x)\), reference density \(r\)
\Ensure Shift estimate(s), de-shifted residual(s), and optionally a cleaned CDT signal or an aligned average

\State Compute the CDT of the template:
$
\widehat u_0(\alpha)
$

\State For each observation, compute the CDT:
$
\widehat u_{\mathrm{obs}}(t,\alpha)
$

\State For each observation, estimate the shift by projection onto the constant mode:
\[
s^\ast(t)=\overline{\widehat u_{\mathrm{obs}}(t,\cdot)}-\overline{\widehat u_0}
\]

\State For each observation, form the de-shifted residual:
\[
\rho(t,\alpha)=\widehat u_{\mathrm{obs}}(t,\alpha)-\widehat u_0(\alpha)-s^\ast(t)
\]

\If{only a single observation is available}
    \State Optionally denoise or smooth \(\rho(t,\cdot)\) to obtain a cleaned residual: 
    $
    \rho_{\mathrm{den}}(t,\alpha)
    $
    \State Form the cleaned CDT signal
    \[
    \widehat u_{\mathrm{clean}}(t,\alpha)=\widehat u_0(\alpha)+s^\ast(t)+\rho_{\mathrm{den}}(t,\alpha)
    \]
    \State Optionally invert the CDT to recover a cleaned physical-space signal
\Else
    \State De-shift each observation:
    \[
    \widehat u_{\mathrm{obs}}^{\mathrm{aligned}}(t,\alpha)
    =
    \widehat u_{\mathrm{obs}}(t,\alpha)-s^\ast(t)
    \]
    \State Average the aligned CDT signals:
    \[
    \widehat u_{\mathrm{avg}}(\alpha)
    =
    \frac{1}{N}\sum_{k=1}^N \widehat u_{\mathrm{obs},k}^{\mathrm{aligned}}(\alpha)
    \]
    \State Interpret \(\widehat u_{\mathrm{avg}}\) as a refined aligned signal or template estimate
    \State Optionally form cleaned individual samples by re-inserting the estimated shifts:
    \[
    \widehat u_{\mathrm{clean},k}(\alpha)=\widehat u_{\mathrm{avg}}(\alpha)+s_k^\ast
    \]
    \State Optionally invert the CDT to recover cleaned physical-space signals
\EndIf
\end{algorithmic}
\end{algorithm}

\begin{remark}
In the single-observation setting, the natural output is a cleaned residual \(\rho_{\mathrm{den}}\) and the corresponding cleaned CDT signal \(\widehat u_{\mathrm{clean}}\). In the multiple-observation setting, the natural output is instead the aligned average
$
\widehat u_{\mathrm{avg}},
$
which may be interpreted as a refined estimate of the common aligned/template CDT signal. Individual cleaned samples may then be obtained by re-inserting the estimated shifts.
\end{remark}

\subsection{Joint Recovery of Template and Shifts from Multiple Observations}

We now consider the practically relevant case in which the template, shifts, and perturbations are all unknown. A single observation is not sufficient to identify all components, so we assume that multiple observations are available and that they share a common underlying template in CDT space. The conceptual pipeline for this joint recovery is illustrated in \cref{fig:joint_pipeline}.

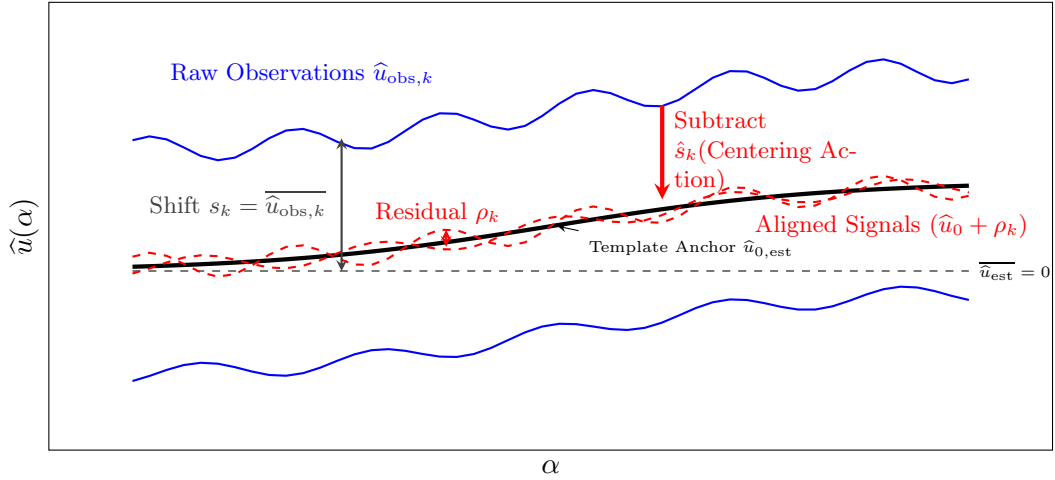
\begin{figure}[h!]
\centering
\begin{tikzpicture}
\begin{axis}[
    width=0.9\textwidth, height=7.5cm,
    title={Joint CDT Recovery Pipeline: From Observations to Template (\cref{alg:joint_cdt_recovery})},
    xlabel={$\alpha$}, ylabel={$\widehat{u}(\alpha)$},
    ticks=none, domain=-3:3, samples=50,
    ymin=-2.0, ymax=3.0, 
    legend style={at={(0.5,-0.2)}, anchor=north, legend columns=-1},
    clip=false
]
    \draw[black, dashed, thin] (axis cs: -3, 0.0) -- (axis cs: 3, 0.0);
    \node at (axis cs: 3.0, 0.0) [black, anchor=west, font=\tiny] {$\overline{\widehat{u}_{\rm est}} = 0$};

    \addplot[black, ultra thick] {1/(1+exp(-x))};
    \node at (axis cs: 0.2, 0.45) [black, anchor=north west, font=\tiny, align=left] {Template Anchor $\widehat{u}_{0,\text{est}}$};
    \draw[->, >=stealth, thin, black] (axis cs: 0.2, 0.45) -- (axis cs: 0.05, 0.5);

    \addplot[blue, opacity=0.5, thick] {1/(1+exp(-x)) + 1.3 + 0.15*sin(deg(6*x))}; 
    \addplot[blue, opacity=0.5, thick] {1/(1+exp(-x)) - 1.2 + 0.1*cos(deg(5*x))};
    \node at (axis cs: -2.8, 2.2) [blue, anchor=west, font=\footnotesize] {Raw Observations $\widehat{u}_{\mathrm{obs},k}$};

    \addplot[red, thick, dashed] {1/(1+exp(-x)) + 0.15*sin(deg(6*x))};
    \addplot[red, thick, dashed] {1/(1+exp(-x)) + 0.1*cos(deg(5*x))};
    \node at (axis cs: 3.5, 0.5) [red, anchor=east, font=\footnotesize] {Aligned Signals ($\widehat{u}_0 + \rho_k$)};


    
    \coordinate (BaseRef) at (axis cs: -1.5, 0);       
    \coordinate (Obs1Ref) at (axis cs: -1.5, 1.48);    
    \draw[<->, >=stealth, darkgray, thick] (BaseRef) -- (Obs1Ref);
    \node[darkgray, anchor=east, font=\footnotesize, xshift=-2pt] at (axis cs: -1.5, 0.74) {Shift $s_k = \overline{\widehat{u}_{\mathrm{obs},k}}$};

    \coordinate (ActionStart) at (axis cs: 0.8, 1.85); 
    \coordinate (ActionEnd)   at (axis cs: 0.8, 0.8);  
    \draw[->, >=stealth, ultra thick, red] (ActionStart) -- (ActionEnd) node[midway, right, font=\footnotesize, text width=2.5cm] {Subtract $\hat{s}_k$(Centering Action)};

    \coordinate (AlignedPoint)  at (axis cs: -0.75, 0.46); 
    \coordinate (TemplatePoint) at (axis cs: -0.75, 0.27); 
    \draw[<->, >=stealth, thick, red] (AlignedPoint) -- (TemplatePoint);
    \node[red, anchor=east, font=\footnotesize, xshift=-2pt] at (axis cs: -0.25, 0.65) {Residual $\rho_k$};

\end{axis}
\end{tikzpicture}
\caption{Conceptual pipeline for joint recovery in CDT space. The centered gauge \(\overline{\widehat u_{0,\mathrm{est}}}=0\) fixes the additive ambiguity between the template and the shifts. Raw observations are first de-shifted using their estimated constant modes, producing an aligned family \(\widehat u_0+\rho_k\). Averaging the aligned signals suppresses the residual fluctuations and reveals the recovered template.}
\label{fig:joint_pipeline}
\end{figure}

Let
\[
\widehat u_{\mathrm{obs},k}(\alpha)
=
\widehat u_0(\alpha)+s_k+\rho_k(\alpha),
\qquad
k=1,\dots,N,
\]
where \(\widehat u_0(\alpha)\) is the unknown common template, \(s_k\) are unknown shifts, and \(\rho_k\) are unknown residuals. 

To separate the transport mode from the residual, we impose the centering condition
\begin{equation}\label{eq:rho_centered_multi}
\overline{\rho_k}=0,
\qquad
k=1,\dots,N.
\end{equation}
As illustrated in \cref{fig:joint_pipeline}, this removes the ambiguity between the constant transport mode and the shape fluctuations. The same figure also provides a visual interpretation of \cref{alg:joint_cdt_recovery}.

Even under \eqref{eq:rho_centered_multi}, the pair \((\widehat u_0,\{s_k\})\) is only defined up to a common additive constant:
\[
\widehat u_0\mapsto \widehat u_0+c,
\qquad
s_k\mapsto s_k-c.
\]
We remove this gauge freedom by requiring
\begin{equation}\label{eq:gauge_zero_mean_template}
\overline{\widehat u_0}=0.
\end{equation}
In computation, the true template \(\widehat u_0\) is unknown, so the same gauge convention is enforced on the recovered estimate by requiring \(\overline{\widehat u_{0,\mathrm{est}}}=0\).

\begin{proposition}[Recovery under centered residuals and fixed gauge]
\label{prop:joint_recovery}
Assume
\[
\widehat u_{\mathrm{obs},k}(\alpha)
=
\widehat u_0(\alpha)+s_k+\rho_k(\alpha),
\qquad
\overline{\rho_k}=0,
\qquad
\overline{\widehat u_0}=0.
\]
Then
\[
\overline{\widehat u_{\mathrm{obs},k}}=s_k,
\qquad
k=1,\dots,N,
\]
and therefore the shifts are recovered by
\[
s_k=\overline{\widehat u_{\mathrm{obs},k}}.
\]
Moreover,
\[
\widehat u_0(\alpha)
=
\frac1N\sum_{k=1}^N
\Bigl(\widehat u_{\mathrm{obs},k}(\alpha)-\overline{\widehat u_{\mathrm{obs},k}}\Bigr)
-\frac1N\sum_{k=1}^N \rho_k(\alpha).
\]
In particular, if the residuals average out,
\[
\frac1N\sum_{k=1}^N \rho_k(\alpha)\approx 0,
\]
then
\[
\widehat u_0(\alpha)
\approx
\frac1N\sum_{k=1}^N
\Bigl(\widehat u_{\mathrm{obs},k}(\alpha)-\overline{\widehat u_{\mathrm{obs},k}}\Bigr).
\]
\end{proposition}

\begin{proof}
Taking the \(r\)-weighted mean of
\[
\widehat u_{\mathrm{obs},k}(\alpha)
=
\widehat u_0(\alpha)+s_k+\rho_k(\alpha)
\]
gives
\[
\overline{\widehat u_{\mathrm{obs},k}}
=
\overline{\widehat u_0}+s_k+\overline{\rho_k}.
\]
By \eqref{eq:rho_centered_multi} and \eqref{eq:gauge_zero_mean_template},
\[
\overline{\widehat u_{\mathrm{obs},k}}=s_k.
\]
Subtracting this identity from the observation model yields
\[
\widehat u_{\mathrm{obs},k}(\alpha)-\overline{\widehat u_{\mathrm{obs},k}}
=
\widehat u_0(\alpha)+\rho_k(\alpha).
\]
Averaging over \(k\) gives
\[
\frac1N\sum_{k=1}^N
\Bigl(\widehat u_{\mathrm{obs},k}(\alpha)-\overline{\widehat u_{\mathrm{obs},k}}\Bigr)
=
\widehat u_0(\alpha)+\frac1N\sum_{k=1}^N \rho_k(\alpha),
\]
which proves the result.
\end{proof}

\begin{remark}[Interpretation]
Under the centering condition \(\overline{\rho_k}=0\) and the gauge choice \(\overline{\widehat u_0}=0\), the constant mode of each observation determines the shift, while averaging the de-shifted CDT signals determines the common template up to the average residual. Thus the unknown-template problem reduces to two operations in CDT space:
\begin{enumerate}
\item projection onto \(\mathrm{span}\{1\}\) to estimate shifts,
\item averaging of de-shifted signals to estimate the template.
\end{enumerate}
\end{remark}

The practical recovery procedure mirrors the preceding proposition. One first estimates the shifts from the constant mode of each observed CDT signal, then de-shifts the observations, averages the aligned CDT signals to obtain a preliminary template estimate, and finally enforces the gauge condition by centering the template. Optional denoising may be applied to the aligned average before the final template is fixed.

\begin{algorithm}[h!]
\caption{Joint CDT Recovery of Template and Shifts}
\label{alg:joint_cdt_recovery}
\begin{algorithmic}[1]
\Require Observed signals \(u_{\mathrm{obs},k}(x)\), \(k=1,\dots,N\), reference density \(r\)
\Ensure Estimated template \(\widehat u_{0,\mathrm{est}}\), estimated shifts \(\widehat s_k\), and residuals \(\widehat\rho_k\)

\State Compute the CDT of each observation:
$
\widehat u_{\mathrm{obs},k}(\alpha),
\qquad
k=1,\dots,N
$

\State Estimate the shifts from the constant mode:
$
\widehat s_k=\overline{\widehat u_{\mathrm{obs},k}},
\qquad
k=1,\dots,N
$

\State Form the de-shifted CDT signals:
\[
\widehat u_{\mathrm{obs},k}^{\mathrm{aligned}}(\alpha)
=
\widehat u_{\mathrm{obs},k}(\alpha)-\widehat s_k
\]

\State Form the aligned average:
\[
\widehat u_{0,\mathrm{raw}}(\alpha)
=
\frac1N\sum_{k=1}^N \widehat u_{\mathrm{obs},k}^{\mathrm{aligned}}(\alpha)
\]

\State Optionally denoise or smooth \(\widehat u_{0,\mathrm{raw}}\) to obtain the template estimate: 
$
\widehat u_{0,\mathrm{est}}(\alpha) .
$
If no denoising is performed, set
\[
\widehat u_{0,\mathrm{est}}(\alpha)=\widehat u_{0,\mathrm{raw}}(\alpha)
\]

\State Enforce the gauge \(\overline{\widehat u_{0,\mathrm{est}}}=0\) by centering:
\[
c:=\overline{\widehat u_{0,\mathrm{est}}},
\qquad
\widehat u_{0,\mathrm{est}}(\alpha)\leftarrow \widehat u_{0,\mathrm{est}}(\alpha)-c,
\qquad
\widehat s_k\leftarrow \widehat s_k+c
\]

\State Form the residuals:
\[
\widehat\rho_k(\alpha)
=
\widehat u_{\mathrm{obs},k}(\alpha)-\widehat u_{0,\mathrm{est}}(\alpha)-\widehat s_k
\]

\State Optionally iterate the previous steps if template refinement is desired
\end{algorithmic}
\end{algorithm}

\begin{remark}[Role of the Reference Density]
The weighted means \(\overline{\widehat u_{\mathrm{obs},k}}\) and the orthogonality condition on the residuals are defined relative to the chosen reference density \(r\). Consequently, the numerical values of the estimated shifts and the recovered template depend on this choice. In practice, \(r\) should be chosen so that it captures the relevant region of support of the observed family and provides a meaningful notion of the constant transport mode across the dataset.
\end{remark}

\begin{algorithm}[t]
\caption{Joint CDT Recovery of Template and Shifts}
\label{alg:joint_cdt_recovery}
\begin{algorithmic}[1]
\Require Observed signals \(u_{\mathrm{obs},k}(x)\), \(k=1,\dots,N\), reference density \(r\)
\Ensure Estimated template \(\widehat u_{0,\mathrm{est}}\), estimated shifts \(\widehat s_k\), and residuals \(\widehat\rho_k\)

\State Compute the CDT of each observation:
\[
\widehat u_{\mathrm{obs},k}(\alpha),
\qquad
k=1,\dots,N
\]

\State Estimate the shifts from the constant mode:
\[
\widehat s_k=\overline{\widehat u_{\mathrm{obs},k}},
\qquad
k=1,\dots,N
\]

\State Form the de-shifted CDT signals:
\[
\widehat u_{\mathrm{obs},k}^{\mathrm{aligned}}(\alpha)
=
\widehat u_{\mathrm{obs},k}(\alpha)-\widehat s_k
\]

\State Form the aligned average:
\[
\widehat u_{0,\mathrm{raw}}(\alpha)
=
\frac1N\sum_{k=1}^N \widehat u_{\mathrm{obs},k}^{\mathrm{aligned}}(\alpha)
\]

\State Optionally denoise or smooth \(\widehat u_{0,\mathrm{raw}}\) to obtain the template estimate:
\[
\widehat u_{0,\mathrm{est}}(\alpha)
\]
If no denoising is performed, set
\[
\widehat u_{0,\mathrm{est}}(\alpha)=\widehat u_{0,\mathrm{raw}}(\alpha)
\]

\State Enforce the gauge \(\overline{\widehat u_{0,\mathrm{est}}}=0\) by centering:
\[
c:=\overline{\widehat u_{0,\mathrm{est}}},
\qquad
\widehat u_{0,\mathrm{est}}(\alpha)\leftarrow \widehat u_{0,\mathrm{est}}(\alpha)-c,
\qquad
\widehat s_k\leftarrow \widehat s_k+c
\]

\State Form the residuals:
\[
\widehat\rho_k(\alpha)
=
\widehat u_{\mathrm{obs},k}(\alpha)-\widehat u_{0,\mathrm{est}}(\alpha)-\widehat s_k
\]

\State Optionally invert the final template estimate \(\widehat u_{0,\mathrm{est}}\) and/or the aligned CDT signals \(\widehat u_{\mathrm{obs},k}^{\mathrm{aligned}}\) to recover a physical-space template estimate and aligned physical-space signals

\State Optionally iterate the previous steps if template refinement is desired
\end{algorithmic}
\end{algorithm}

\begin{remark}
The centering step in \cref{alg:joint_cdt_recovery} does not require prior knowledge of \(\widehat u_0\). It enforces the computational gauge \(\overline{\widehat u_{0,\mathrm{est}}}=0\), which is the algorithmic analogue of the theoretical convention \(\overline{\widehat u_0}=0\), while compensating the estimated shifts by the same constant.
\end{remark}

\section{SCDT Recovery for Signed Signals}
\label{sec:scdt}

The CDT-based recovery procedures developed above are tailored to nonnegative normalized densities. For signed signals, these assumptions no longer apply, and it is natural to work instead with the signed cumulative distribution transform (SCDT). The SCDT represents a signal by separating its positive and negative parts and applying transport-based coordinates to each part. In this way, it extends the transport viewpoint of the CDT to signals that change sign.

More precisely, let \(f\) be a signed signal, and write
\[
f=f^+-f^-,
\qquad
f^\pm\ge 0,
\qquad
f^+f^-=0,
\]
where \(f^+\) and \(f^-\) denote the positive and negative parts of \(f\). The SCDT feature map, denoted by
\[
\mathcal S(f),
\]
is obtained by applying CDT-type transport coordinates to the normalized positive and negative components together with their associated masses. Thus \(\mathcal S(f)\) should be understood as a transport feature representation of the signed signal \(f\), built from the geometry of \(f^+\) and \(f^-\).

In the density setting, translation becomes an exact additive constant mode in CDT space. In the signed-signal setting, no comparably simple closed-form formula is available. Nevertheless, translated signals often remain well organized in SCDT feature space, which motivates the use of SCDT for shift estimation and alignment. The procedures below should therefore be viewed as numerical SCDT-based analogues of the CDT recovery algorithms: shifts are estimated by matching SCDT features, the signals are aligned in physical space, and aligned averages are used to recover a common template.

\subsection{Known-Template Recovery}

Let \(f_0(x)\) be a known signed template, and suppose that the observed signals satisfy
\[
f_{\mathrm{obs},k}(x)=f_0(x-s_k)+\eta_k(x),
\qquad
k=1,\dots,N,
\]
where \(s_k\in\R\) are unknown shifts and \(\eta_k\) denotes additive noise. A natural estimator of the shift is obtained by matching each observed signal to the known template over a discrete candidate shift grid \(\mathcal G\subset\R\).

For each observation, we define
\[
\widehat s_k
=
\arg\min_{s\in\mathcal G}
\left\|
\mathcal S\bigl(f_{\mathrm{obs},k}(\cdot+s)\bigr)-\mathcal S(f_0)
\right\|,
\]
and then align the signal in physical space by
\[
f_k^{\mathrm{aligned}}(x)=f_{\mathrm{obs},k}(x+\widehat s_k).
\]

\begin{algorithm}[h!]
\caption{SCDT-Based Shift Recovery with Known Template}
\label{alg:scdt_known}
\begin{algorithmic}[1]
\Require Observed signed signals \(f_{\mathrm{obs},k}(x)\), known template \(f_0(x)\), candidate shift grid \(\mathcal G\)
\Ensure Estimated shifts \(\widehat s_k\) and aligned signals \(f_k^{\mathrm{aligned}}\)

\State Compute the SCDT feature of the template:
$\mathcal S(f_0)$

\For{each observation \(k=1,\dots,N\)}
    \For{each candidate shift \(s\in\mathcal G\)}
        \State Compute the mismatch
        \[
        J_k(s)=
        \left\|
        \mathcal S\bigl(f_{\mathrm{obs},k}(\cdot+s)\bigr)-\mathcal S(f_0)
        \right\|
        \]
    \EndFor
    \State Estimate the shift by discrete search:
    \[
    \widehat s_k=\arg\min_{s\in\mathcal G} J_k(s)
    \]
    \State Align the signal in physical space:
    \[
    f_k^{\mathrm{aligned}}(x)=f_{\mathrm{obs},k}(x+\widehat s_k)
    \]
\EndFor
\end{algorithmic}
\end{algorithm}

\subsection{Unknown-Template Recovery}

We next consider the case in which the common signed template is also unknown. We again assume a family of translated noisy observations
\[
f_{\mathrm{obs},k}(x)=f_0(x-s_k)+\eta_k(x),
\qquad
k=1,\dots,N,
\]
but now both \(f_0\) and \(s_k\) must be estimated from the data. In the absence of a closed-form projection formula, we use an alternating align-and-average procedure.

Starting from an initial template estimate \(f_{\mathrm{temp}}^{(0)}\), we repeat the following steps. At iteration \(m\), shifts are estimated by matching each observation to the current template in SCDT feature space:
\[
\widehat s_k^{(m)}
=
\arg\min_{s\in\mathcal G}
\left\|
\mathcal S\bigl(f_{\mathrm{obs},k}(\cdot+s)\bigr)-\mathcal S\bigl(f_{\mathrm{temp}}^{(m)}\bigr)
\right\|.
\]
The observations are then aligned in physical space by
\[
f_{k,\mathrm{aligned}}^{(m)}(x)=f_{\mathrm{obs},k}(x+\widehat s_k^{(m)}),
\]
and the template is updated by averaging:
\[
f_{\mathrm{temp}}^{(m+1)}(x)=\frac1N\sum_{k=1}^N f_{k,\mathrm{aligned}}^{(m)}(x).
\]
Unlike the CDT density setting, the SCDT recovery procedures do not yield an explicit shift formula. Instead, the shifts are approximated numerically by minimizing an SCDT feature mismatch over a prescribed candidate shift grid \(\mathcal G\). The quality of recovery therefore depends in part on choosing \(\mathcal G\) wide enough to contain the expected translations and fine enough to resolve the minimizer.

In practice, the initial template \(f_{\mathrm{temp}}^{(0)}\) may be chosen as the direct average of the observed signals,
\[
f_{\mathrm{temp}}^{(0)}(x)
=
\frac1N\sum_{k=1}^N f_{\mathrm{obs},k}(x),
\]
or by any other coarse representative of the dataset. In the numerical experiments below, we use the direct average as the initialization.

\begin{algorithm}[h!]
\caption{SCDT-Based Joint Recovery with Unknown Template}
\label{alg:scdt_unknown}
\begin{algorithmic}[1]
\Require Observed signed signals \(f_{\mathrm{obs},k}(x)\), candidate shift grid \(\mathcal G\), initial template estimate \(f_{\mathrm{temp}}^{(0)}\), number of iterations \(M\)
\Ensure Estimated template \(f_{\mathrm{temp}}^{(M)}\), estimated shifts \(\widehat s_k^{(M-1)}\), and aligned signals

\For{\(m=0,\dots,M-1\)}
    \State Compute the SCDT feature of the current template:
    $
    \mathcal S\bigl(f_{\mathrm{temp}}^{(m)}\bigr)
    $ 

    \For{each observation \(k=1,\dots,N\)}
        \For{each candidate shift \(s\in\mathcal G\)}
            \State Compute the mismatch
            \[
            J_k^{(m)}(s)=
            \left\|
            \mathcal S\bigl(f_{\mathrm{obs},k}(\cdot+s)\bigr)
            -
            \mathcal S\bigl(f_{\mathrm{temp}}^{(m)}\bigr)
            \right\|
            \]
        \EndFor
        \State Estimate the shift by discrete search:
        \[
        \widehat s_k^{(m)}=\arg\min_{s\in\mathcal G} J_k^{(m)}(s)
        \]
        \State Align the signal in physical space:
        \[
        f_{k,\mathrm{aligned}}^{(m)}(x)=f_{\mathrm{obs},k}(x+\widehat s_k^{(m)})
        \]
    \EndFor

    \State Update the template by averaging aligned signals:
    $
    f_{\mathrm{temp}}^{(m+1)}(x)=\frac1N\sum_{k=1}^N f_{k,\mathrm{aligned}}^{(m)}(x)
    $
\EndFor
\end{algorithmic}
\end{algorithm}

\begin{remark}
The output of the SCDT algorithms is already expressed in physical space through the aligned signals \(f_k^{\mathrm{aligned}}\) and, in the unknown-template case, the averaged template estimate \(f_{\mathrm{temp}}^{(M)}\). Thus no inverse-transform post-processing is required in the way it was for the CDT density experiments.
\end{remark}

\section{Numerics}
\label{sec:numerics}

We organize the numerical experiments in five parts. We first illustrate the geometric linearization of transport-dominated variability under the CDT. We then validate the first-order CDT perturbation formula in a semi-analytic setting. Next, we study CDT-based shift and template recovery in the known-template and unknown-template settings. Finally, we examine SCDT-based recovery for translated signed signals.

\subsection{Transport Linearization}

Before turning to perturbation effects, we first illustrate the geometric advantage of the CDT for transport-dominated data. We consider a translated Gaussian-mixture family with Gaussian reference density
\[
r(\alpha)=\mathcal N(0,2.5^2).
\]
The translated snapshots are generated numerically on a finite grid by interpolation and renormalization, and the CDT is computed by numerical CDF inversion. In physical space, the translated snapshots remain strongly nonlinear as functions of the spatial variable, whereas in CDT space the transformed snapshots form an almost affine family.

To quantify this, we compare the computed CDT snapshots with the ideal affine model
\[
\widehat u_s(\alpha)=\widehat u_0(\alpha)+s,
\]
and measure the relative translation-affine error
\[
\frac{
\left(
\sum_{j=1}^{31}
\|\widehat u_{s_j}-(\widehat u_0+s_j)\|_2^2
\right)^{1/2}
}{
\left(
\sum_{j=1}^{31}
\|\widehat u_{s_j}\|_2^2
\right)^{1/2}
},
\]
where \(s_1,\dots,s_{31}\) are sampled uniformly from \([-1.5,1.5]\), and the norm is evaluated numerically on the sampled \(\alpha\)-grid. The resulting relative translation-affine error is approximately \(1.16\times 10^{-2}\), indicating that the transformed snapshots lie very close to an affine line in CDT space.

\begin{figure}[h!]
    \centering
    \includegraphics[width=\textwidth]{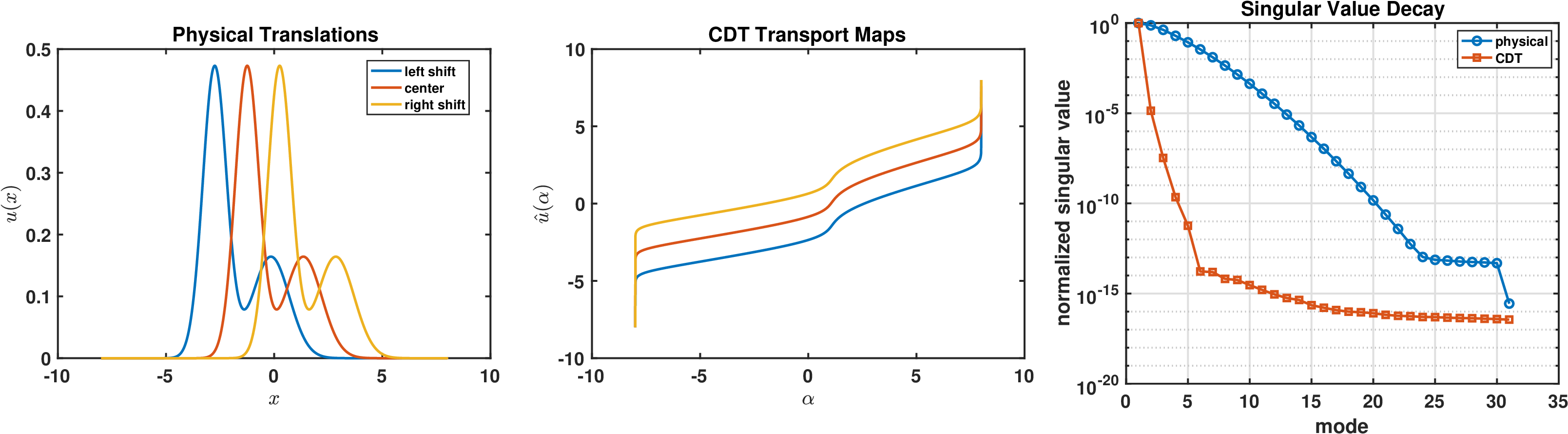}
    \caption{Geometric effect of the CDT on a translated Gaussian-mixture family. Left: representative translated densities in physical space. Center: the corresponding CDT transport maps, which are nearly affine under translation. Right: normalized singular values of the mean-centered snapshot matrices in physical space and in CDT space. The CDT representation is substantially more low-dimensional.}
    \label{fig:cdt_geometry}
\end{figure}

This geometric simplification is reflected in the singular value decay of the corresponding mean-centered snapshot matrices (see \cref{fig:cdt_geometry}). In physical space, the first three singular values are
\[
1.434\times 10^{1},\qquad 1.071\times 10^{1},\qquad 6.029,
\]
whereas in CDT space they are
\[
2.227\times 10^{2},\qquad 3.042\times 10^{-3},\qquad 7.367\times 10^{-6}.
\]
Thus the transformed family is dramatically more low-dimensional, which explains why CDT-based representations are effective for transport-dominated variability. This low-dimensional geometric structure provides the baseline against which the effect of additive perturbations can be assessed.

\subsection{Validation of the First-Order Perturbation Formula}

We next validate the first-order CDT perturbation formula (cf.~\cref{thm:cdt_noise}) in a semi-analytic setting where the quantile equation can be solved accurately and the asymptotic scaling can be examined directly. We consider a translated Gaussian density with mean \(0.6\) and unit variance, together with standard Gaussian reference density, so that the unperturbed CDT is known exactly. A smooth zero-mass perturbation is introduced through a primitive \(E(x)\), and, for each \(\alpha\), the perturbed transport map is obtained by numerically solving the quantile equation
\[
U_\delta(x)=U(x)+\delta E(x),
\qquad
U_\delta(\hat u_\delta(\alpha))=R(\alpha).
\]
The first-order theory predicts
\[
\hat u_\delta(\alpha)-\hat u(\alpha)
=
\delta \xi(\alpha)+O(\delta^2),
\quad \mbox{with} \quad
\xi(\alpha):=-\frac{E(\hat u(\alpha))}{u(\hat u(\alpha))}.
\]

\begin{figure}[h!]
    \centering
    \includegraphics[width=0.6\textwidth]{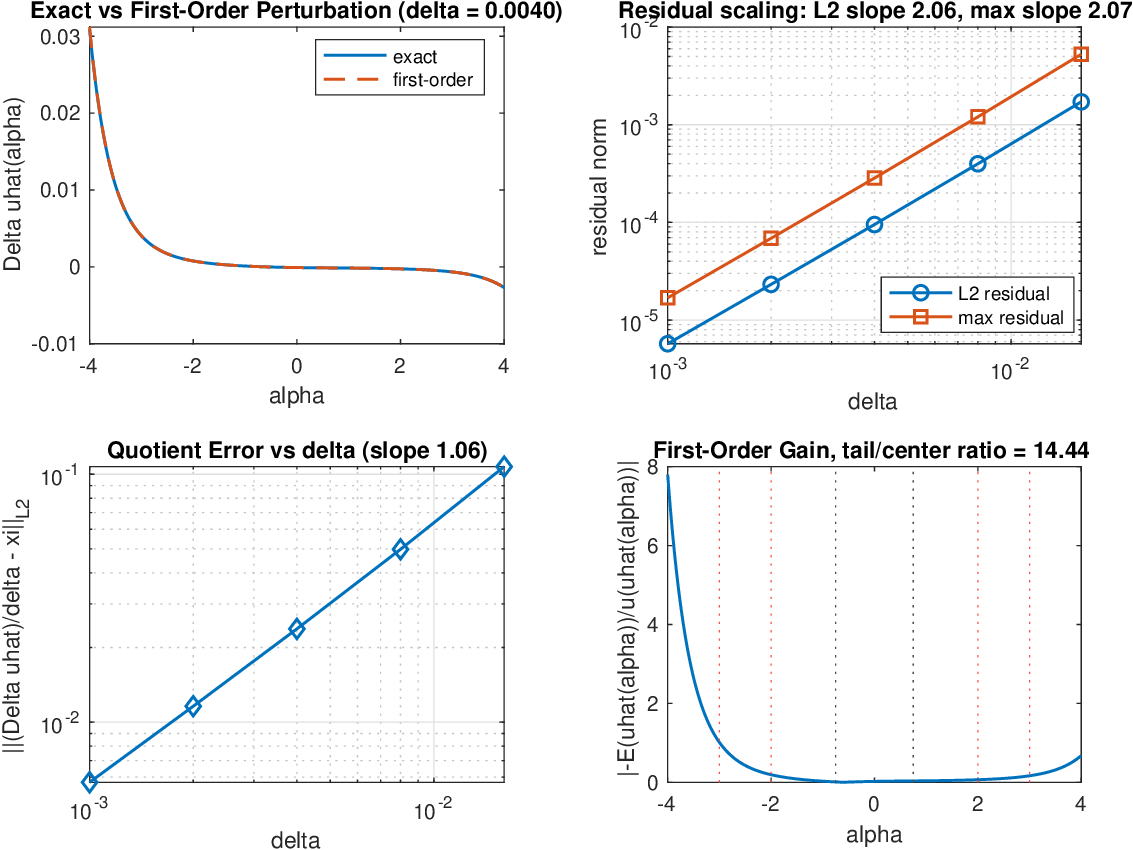}
    \caption{Semi-analytic validation of the first-order CDT perturbation formula. Top left: exact and first-order predicted perturbations for a representative value of \(\delta\). Top right: \(L^2\) and \(L^\infty\) residual norms versus \(\delta\), showing approximately quadratic scaling. Bottom left: quotient error \(\|(\hat u_\delta-\hat u)/\delta-\xi\|_{L^2}\) versus \(\delta\), showing the expected linear scaling. Bottom right: first-order gain factor, illustrating strong tail amplification in low-density regions.}
    \label{fig:cdt_small_noise_analytic}
\end{figure}

Figure~\ref{fig:cdt_small_noise_analytic} compares the exact perturbation with the first-order prediction and reports the scaling of the residual \(\hat u_\delta-\hat u-\delta\xi\) as \(\delta\to 0\). The numerically observed residuals satisfy
\[
\|\hat u_\delta-\hat u-\delta\xi\|_{L^2}=O(\delta^2),
\qquad
\|\hat u_\delta-\hat u-\delta\xi\|_{L^\infty}=O(\delta^2),
\]
with measured slopes \(2.056\) and \(2.070\), respectively. Equivalently, the quotient error \(\|(\hat u_\delta-\hat u)/\delta-\xi\|_{L^2}\) 
scales as \(O(\delta)\), with measured slope \(1.056\). Here all norms are evaluated numerically on the sampled \(\alpha\)-grid. These results give sharp numerical confirmation of the first-order perturbation formula.

The same experiment also illustrates the density-weighted amplification predicted by the theory. The gain factor
$
\left|\frac{E(\hat u(\alpha))}{u(\hat u(\alpha))}\right|
$
is substantially larger in the tails than near the center of the distribution, with a measured tail-to-center ratio of approximately \(14.4\). Here the ratio compares the average magnitude of the gain factor on the tail region \(|\alpha|\in[2,3]\) with its average magnitude on the central region \(|\alpha|\le 0.75\). Thus, even for smooth perturbations, the CDT is markedly more sensitive in low-density regions.

\subsection{Projection-Based Shift Recovery with Known Template}

We first consider the known-template setting of \cref{alg:cdt_deshift}, where the reference template \(u_0\) is assumed available and only the transport shift must be estimated from noisy observations. The experiments use a Gaussian reference density 
$r(\alpha)=\mathcal N(0,2.5^2)$, 
together with a spatial/transform grid of 2001 points on \([-8,8]\). For each template, we generate \(N=21\) observations with evenly spaced shifts
\[
s_k\in[-1,1],
\]
and we report results at SNR levels \(\infty\), \(20\), \(10\), and \(0\) dB, where \(\infty\) dB denotes the noiseless case.

Each observation is constructed from a translated density of the form
\[
u_{\mathrm{obs},k}(x)=u_0(x-s_k)+\varepsilon_k(x),
\]
where \(\varepsilon_k\) is a smooth zero-mass perturbation scaled to the prescribed SNR. To preserve admissibility of the noisy densities, the perturbation is clipped when necessary to maintain nonnegativity, and the result is then renormalized. Figure~\ref{fig:cdt_shift_known} shows the Gaussian-mixture example.

\begin{figure}[h!]
    \centering
    \includegraphics[width=0.9\textwidth]{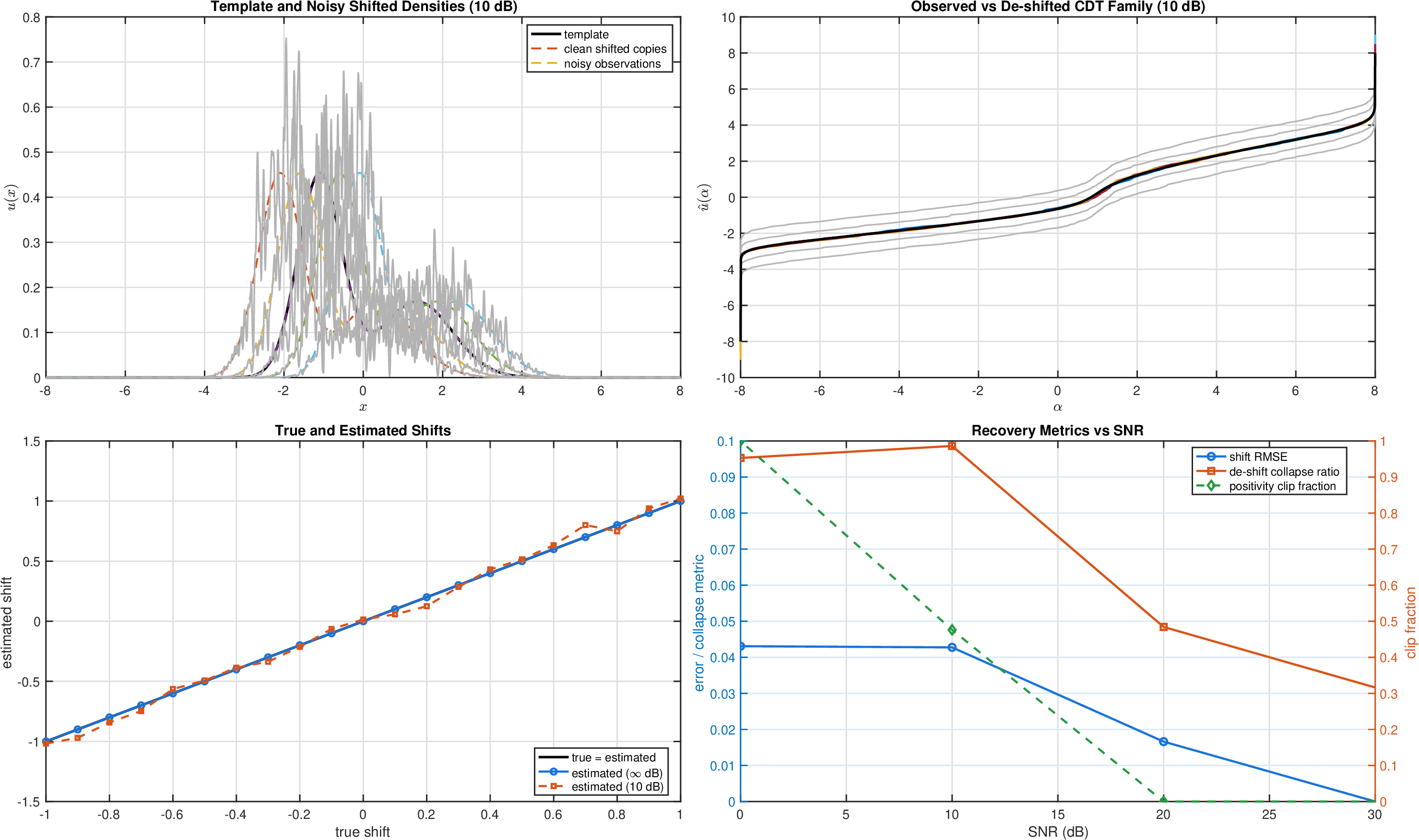}
    \caption{Known-template shift recovery in CDT space for a Gaussian-mixture template. Top left: the common template, several clean shifted copies, and the corresponding noisy observations in physical space. Top right: observed CDT snapshots and the de-shifted CDT family after projection onto the constant mode. Bottom left: true versus estimated shifts. Bottom right: shift RMSE, de-shift collapse ratio, and positivity-clip fraction versus SNR.}
    \label{fig:cdt_shift_known}
\end{figure}

Figure~\ref{fig:cdt_shift_known} illustrates the recovery mechanism for a Gaussian-mixture template. The top-left panel shows the common template, several clean shifted copies, and the corresponding noisy observations in physical space. The top-right panel then shows the observed CDT family together with the de-shifted CDT signals obtained by subtracting the projection-based shift estimates
\[
s_k^\ast=\overline{\widehat u_{\mathrm{obs},k}}-\overline{\widehat u_0}.
\]
As predicted by the theory, the aligned CDT family collapses toward the template after the constant transport mode is removed.

The bottom-left panel compares the true and estimated shifts. In the noiseless case, the recovery is exact up to numerical precision, while at moderate noise levels the estimates remain close to the diagonal. The bottom-right panel summarizes the behavior as the noise level increases. For the Gaussian-mixture example, the shift RMSE remains small from \(\infty\) dB down to \(10\) dB, while the de-shift collapse ratio also stays low, confirming that alignment continues to remove the dominant transport variability. At the strongest noise level, positivity clipping becomes active in order to preserve admissibility of the noisy densities, so this regime should be interpreted as a stress test of the CDT-based density model rather than an unconstrained additive perturbation model.

To visualize the effect of alignment back in physical space, Figure~\ref{fig:cdt_shift_known_inverse} compares the average noisy density
$
\frac1N\sum_{k=1}^N u_{\mathrm{obs},k}(x)
$
with the physical-space signal obtained by first averaging the de-shifted CDT observations and then applying the inverse CDT:
\[
u_{\mathrm{avg}}(x)
=
\mathrm{CDT}^{-1}\!\left(
\frac1N\sum_{k=1}^N \bigl(\widehat u_{\mathrm{obs},k}-s_k^\ast\bigr)
\right).
\]
Direct averaging of the shifted noisy observations visibly blurs the template, whereas averaging after CDT-based alignment produces a substantially sharper physical-space profile. This provides a concrete interpretation of the projection-based de-shifting procedure.

\begin{figure}[h!]
    \centering
    \includegraphics[width=0.45\textwidth]{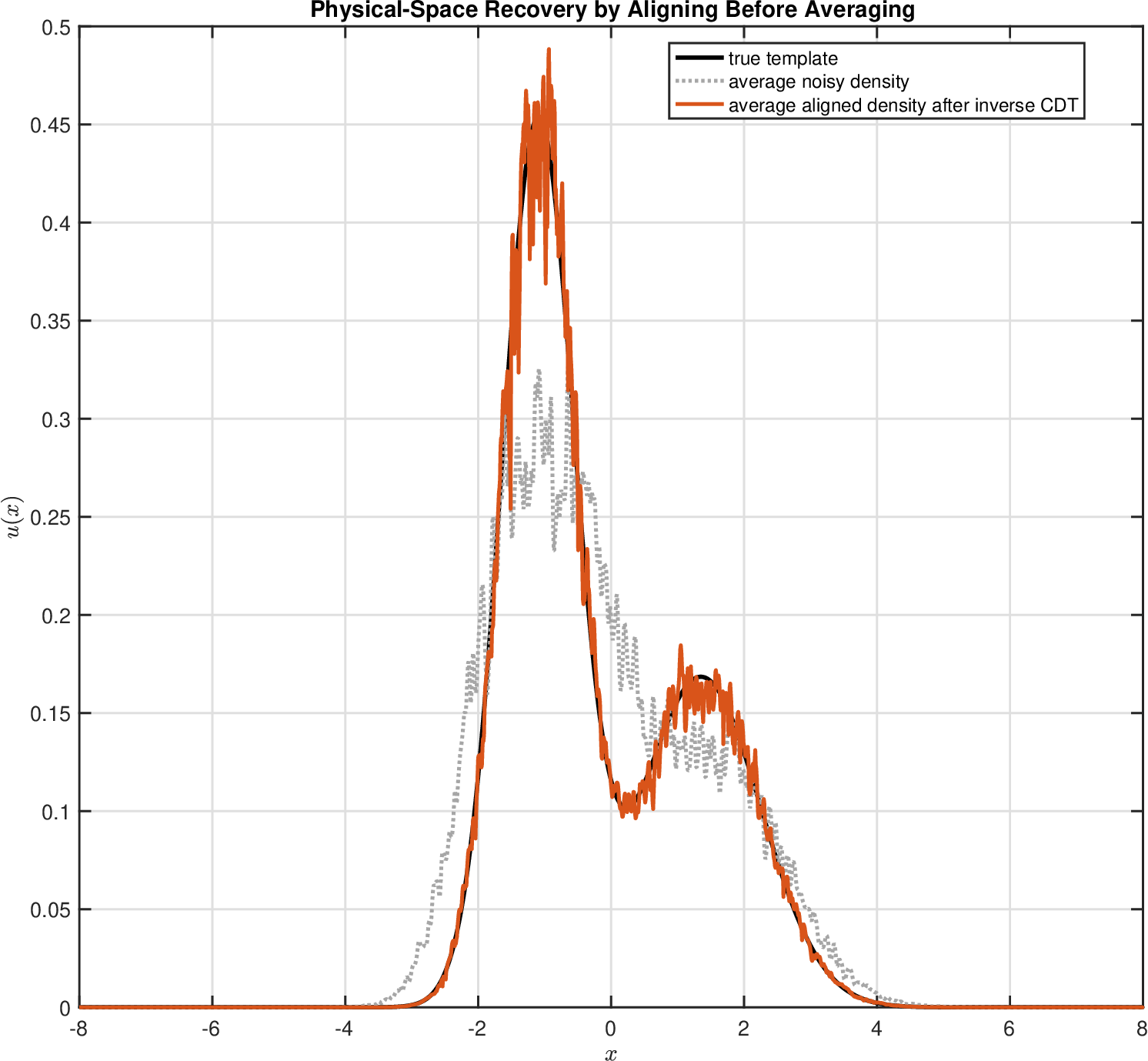}
    \caption{Physical-space interpretation of CDT-based de-shifting in the known-template setting. The black curve is the true template, the gray curve is the direct average of the noisy shifted densities, and the orange curve is obtained by first averaging the de-shifted CDT signals and then applying the inverse CDT. }

    \label{fig:cdt_shift_known_inverse}
\end{figure}

\subsection{Joint Recovery of Template and Shifts from Multiple Observations}

We next consider the unknown-template setting, where neither the common template nor the individual shifts are assumed known a priori. As in the known-template experiments, we use a Gaussian reference density
$
r(\alpha)=\mathcal N(0,2.5^2),
$
together with a spatial/transform grid of 2001 points on \([-8,8]\). For each template, we generate \(N=21\) observations with evenly spaced shifts
\[
s_k\in[-1,1],
\]
and we report results at SNR levels \(\infty\), \(20\), \(10\), and \(0\) dB, where \(\infty\) dB denotes the noiseless case. The observations are generated from translated densities with smooth zero-mass perturbations, followed by positivity clipping and renormalization when needed.

In this setting, we apply the de-shift--and--average strategy of \cref{alg:joint_cdt_recovery}: shifts are estimated from the constant mode of each CDT observation, the family is aligned in CDT space, and the common template is estimated by averaging the aligned signals under the centered gauge \(\overline{\widehat u_{0,\mathrm{est}}}=0\). In the synthetic experiments, the true template is used only to generate the data and to evaluate the recovery error; it is not used by the recovery algorithm itself.

\begin{figure}[h!]
    \centering
    \includegraphics[width=0.8\textwidth]{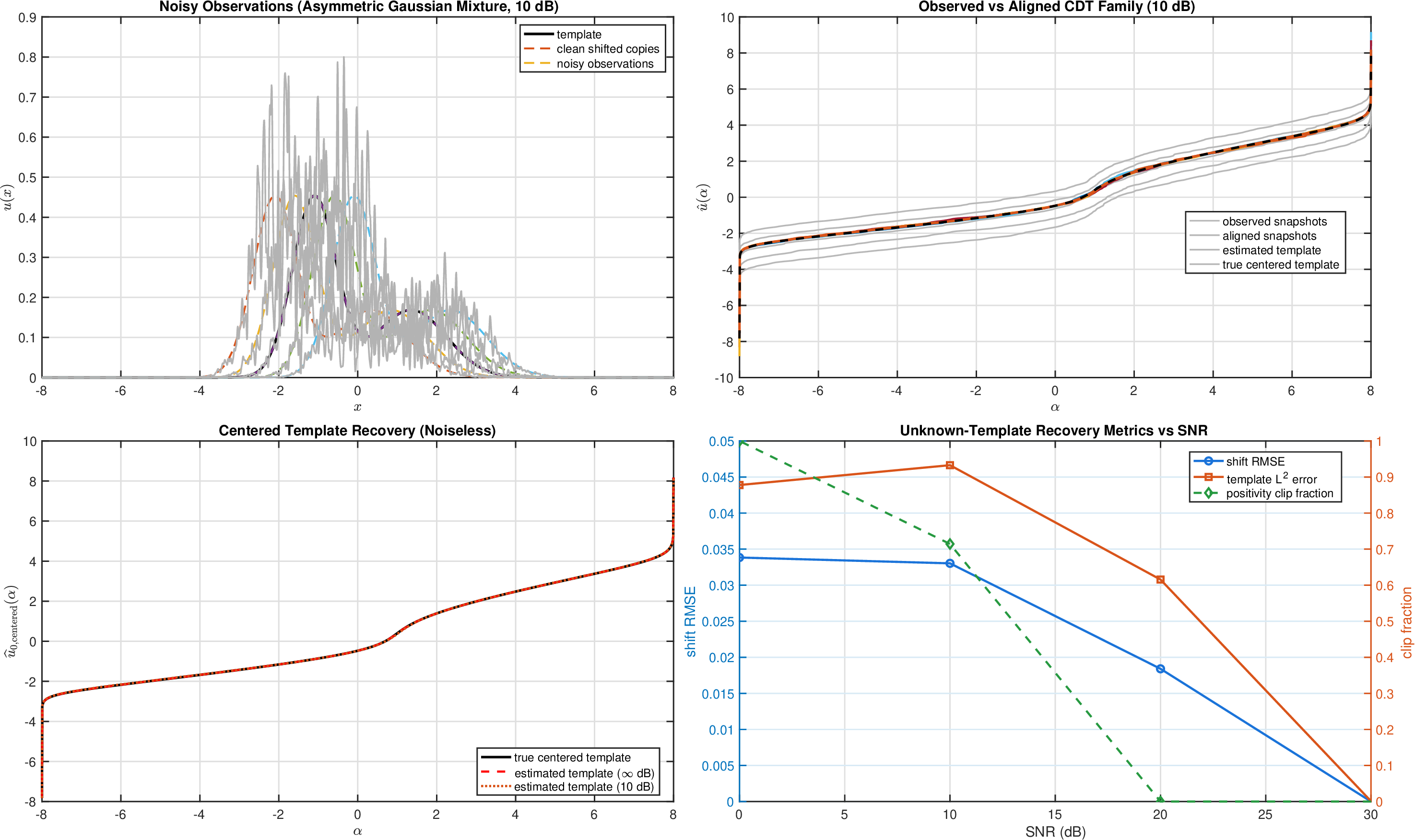}
    \caption{Unknown-template shift and template recovery in CDT space for a Gaussian-mixture family. Top left: the common template, several clean shifted copies, and noisy observations in physical space. Top right: observed CDT snapshots, aligned CDT snapshots, the recovered centered template, and the true centered template. Bottom left: true versus recovered centered template. Bottom right: shift RMSE, template \(L^2\) error, and positivity-clip fraction versus SNR.}
    \label{fig:cdt_shift_unknown}
\end{figure}

Figure~\ref{fig:cdt_shift_unknown} shows the resulting recovery behavior for the same Gaussian-mixture family. The top-left panel displays several noisy observations in physical space together with the underlying template and representative clean shifted copies. The top-right panel shows the observed CDT family, the aligned CDT family after shift removal, the recovered centered template, and the true centered template in CDT space. As in the known-template case, alignment removes the dominant transport variability and produces a markedly more coherent family.

The bottom-left panel compares the true centered template with the recovered template. In the noiseless case, the reconstruction is exact up to numerical precision; under moderate noise, the estimated template remains close to the truth, with the main deviations occurring in the strongest-noise regime. The bottom-right panel summarizes the shift and template errors as functions of the SNR. For the Gaussian-mixture example, both the shift RMSE and the template \(L^2\) error remain moderate down to \(10\) dB, again with positivity clipping becoming significant in the strongest-noise cases.

Figure~\ref{fig:cdt_shift_unknown_inverse} gives the corresponding physical-space interpretation. The direct average of the noisy shifted observations is blurred by transport variability, whereas the inverse transform of the raw aligned CDT average is noticeably closer to the true template. The recovered template obtained from the unknown-template algorithm, after gauge centering in CDT space and inverse transformation, is likewise close to the true physical-space template. Thus the same alignment principle that enables shift estimation also supports meaningful template reconstruction from multiple observations.

\begin{figure}[h!]
    \centering
    \includegraphics[width=0.45\textwidth]{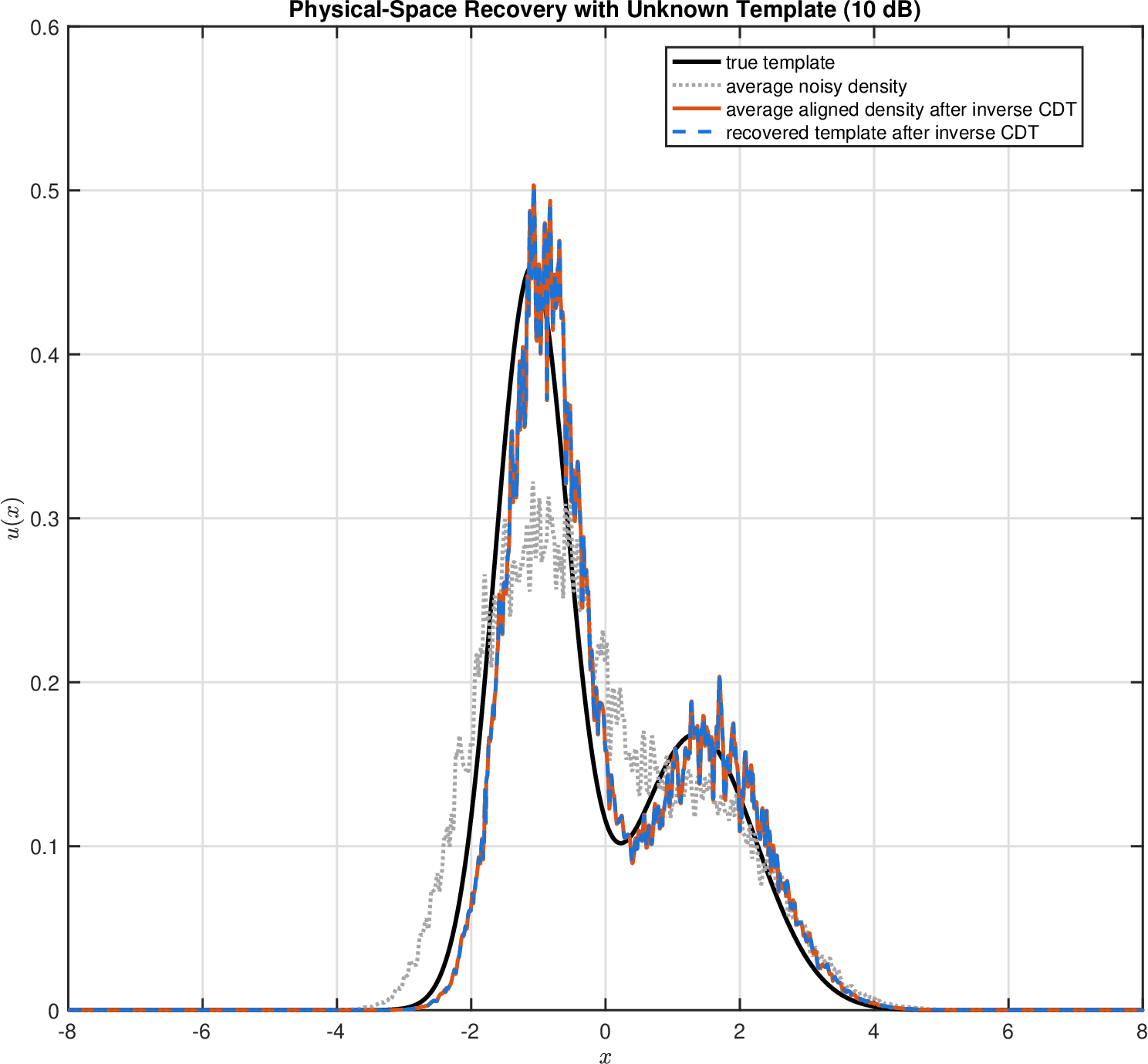}
    \caption{Physical-space interpretation of the unknown-template recovery procedure. The black curve is the true template, the gray curve is the direct average of the noisy shifted observations, the orange curve is obtained by first averaging the aligned CDT signals and then applying the inverse CDT, and the blue dashed curve is the inverse CDT of the final recovered template estimate after gauge centering. Alignment before averaging yields a substantially sharper estimate than direct averaging, and the recovered template remains close to the true physical-space template.}

    \label{fig:cdt_shift_unknown_inverse}
\end{figure}

\subsection{SCDT Recovery for Signed Signals}

We next illustrate the SCDT-based recovery procedures of \cref{alg:scdt_known,alg:scdt_unknown} on translated signed signals with additive noise. We consider three representative signal classes on a uniform grid: a localized oscillatory Gabor signal, an asymmetric sawtooth signal, and a square-wave-type signal with sharp discontinuities. In all cases, translated observations are generated with shifts
\[
s_k\in[-0.13,0.13],
\]
sampled uniformly, and additive noise is added at prescribed SNR levels.

Since the unknown-template setting is the more demanding recovery problem, we focus here on that case. The corresponding known-template experiments exhibit similar alignment, so we omit them for brevity.

Representative unknown-template results for the Gabor signal are shown in \cref{fig:scdt_shift_unknown_gabor}. The template and shifts are recovered jointly by alternating SCDT-based matching and averaging of aligned signals. The dominant translation structure is recovered accurately over a broad SNR range, and the aligned average produces a meaningful template estimate. The corresponding signal-space interpretations for the Gabor, sawtooth, and square signals are shown in \cref{fig:scdt_shift_unknown_gabor_inverse}, where the aligned average is compared with the direct average of the noisy shifted observations.

To compare performance across all three signal classes, \cref{fig:scdt_shift_unknown_summary} reports the shift RMSE and template \(L^2\) error as functions of the SNR. The Gabor and sawtooth signals exhibit the strongest recovery behavior, while the square signal is the most challenging case, especially in the stronger-noise regime.

\begin{figure}[h!]
    \centering
    \includegraphics[width=0.9\textwidth]{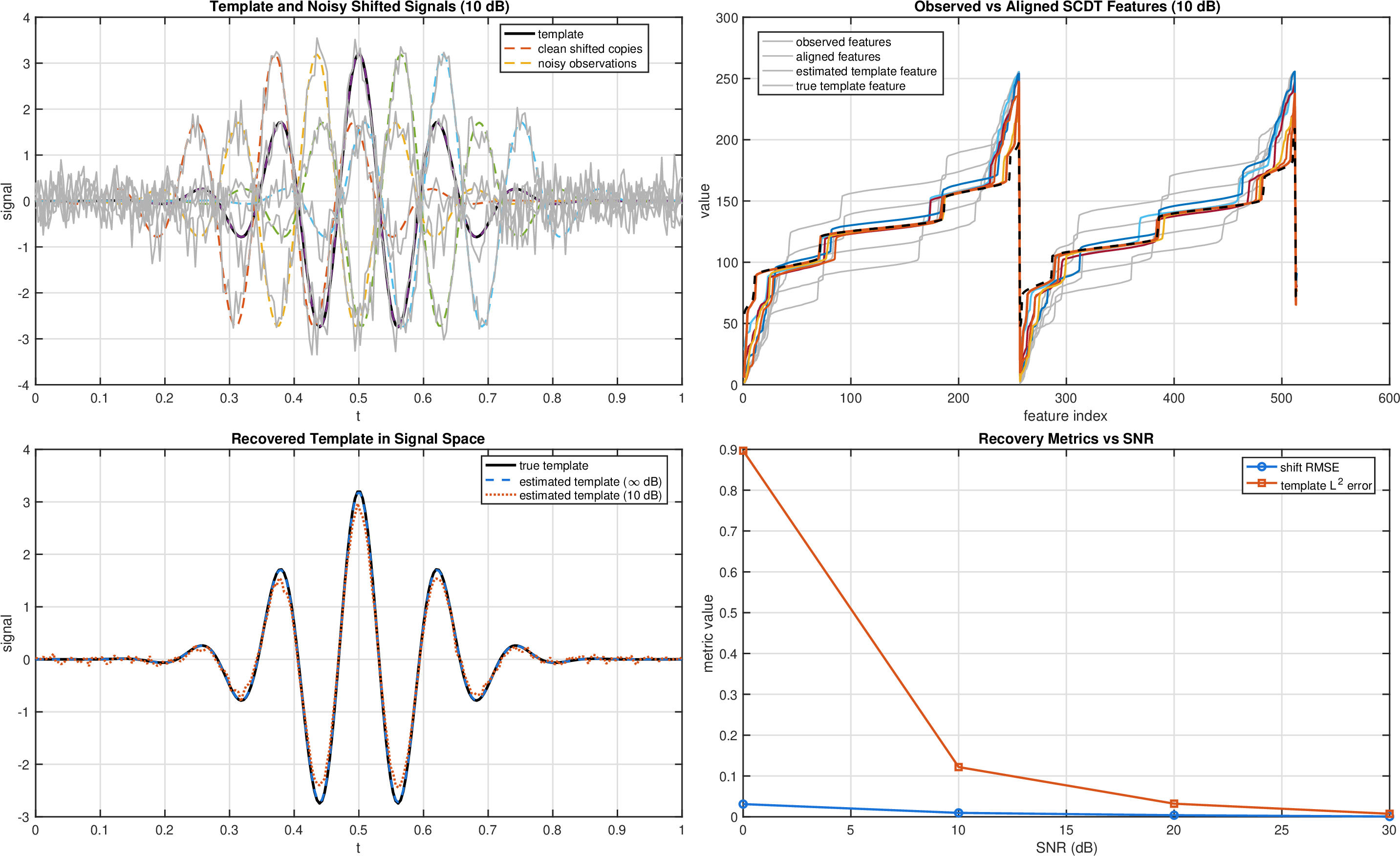}
    \caption{Unknown-template SCDT recovery for a translated Gabor signal. Top left: template, clean shifted copies, and noisy observations in signal space. Top right: observed and aligned SCDT features together with the estimated and true template features. Bottom left: recovered template in signal space. Bottom right: shift RMSE and template \(L^2\) error versus SNR.}
    \label{fig:scdt_shift_unknown_gabor}
\end{figure}

\begin{figure}[h!]
    \centering
    \includegraphics[width=0.32\textwidth]{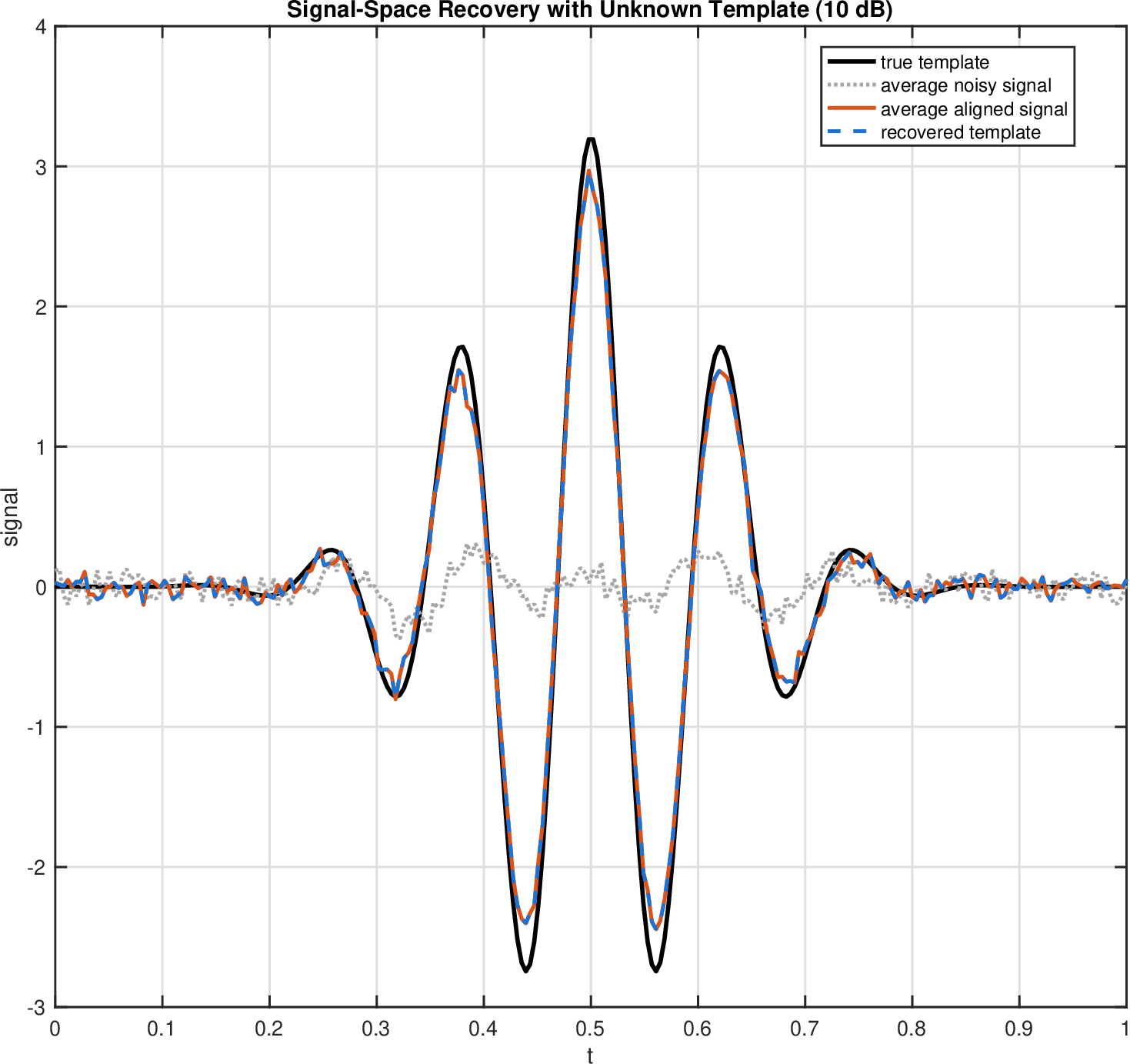}
    \includegraphics[width=0.32\textwidth]{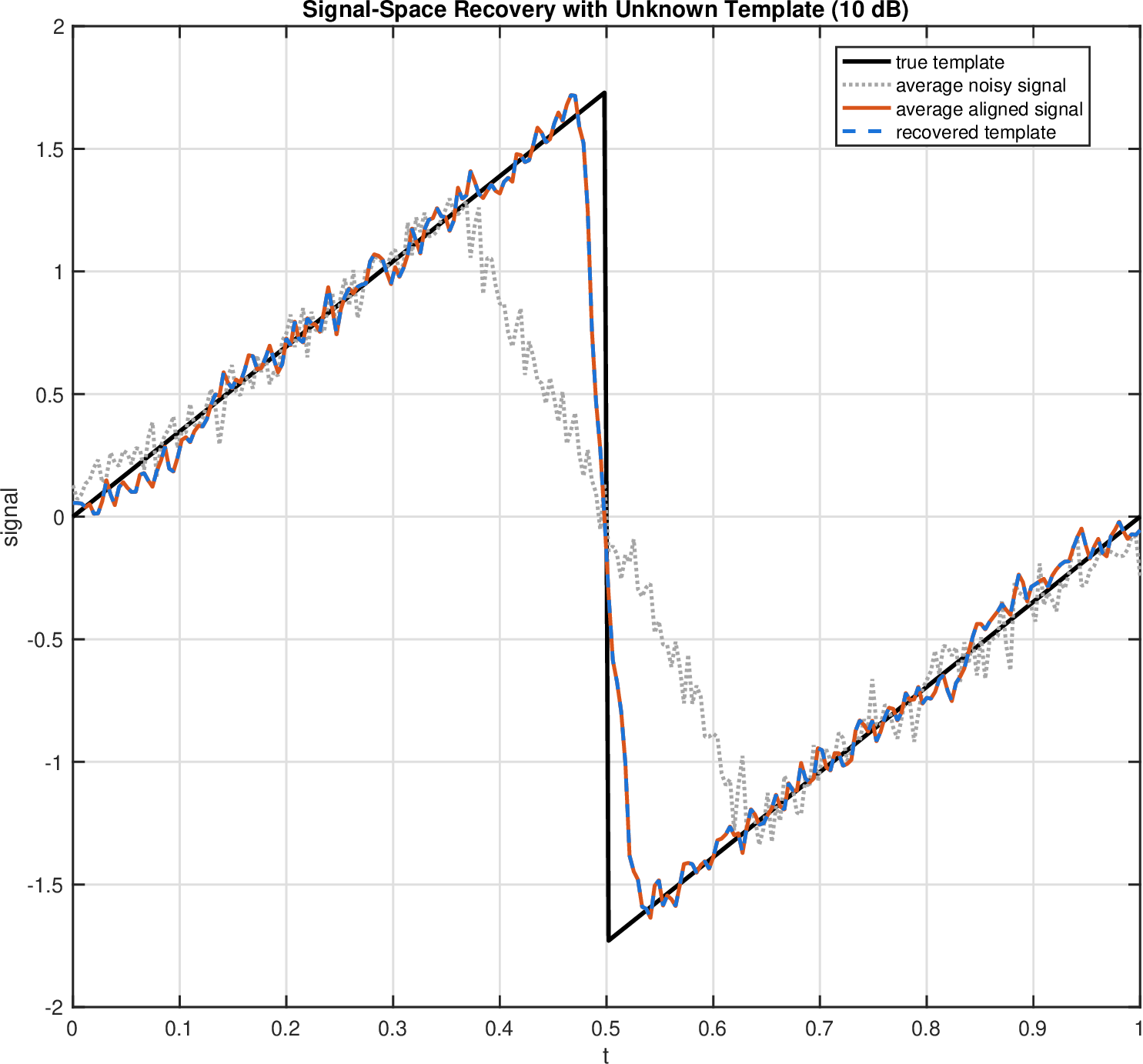}
    \includegraphics[width=0.32\textwidth]{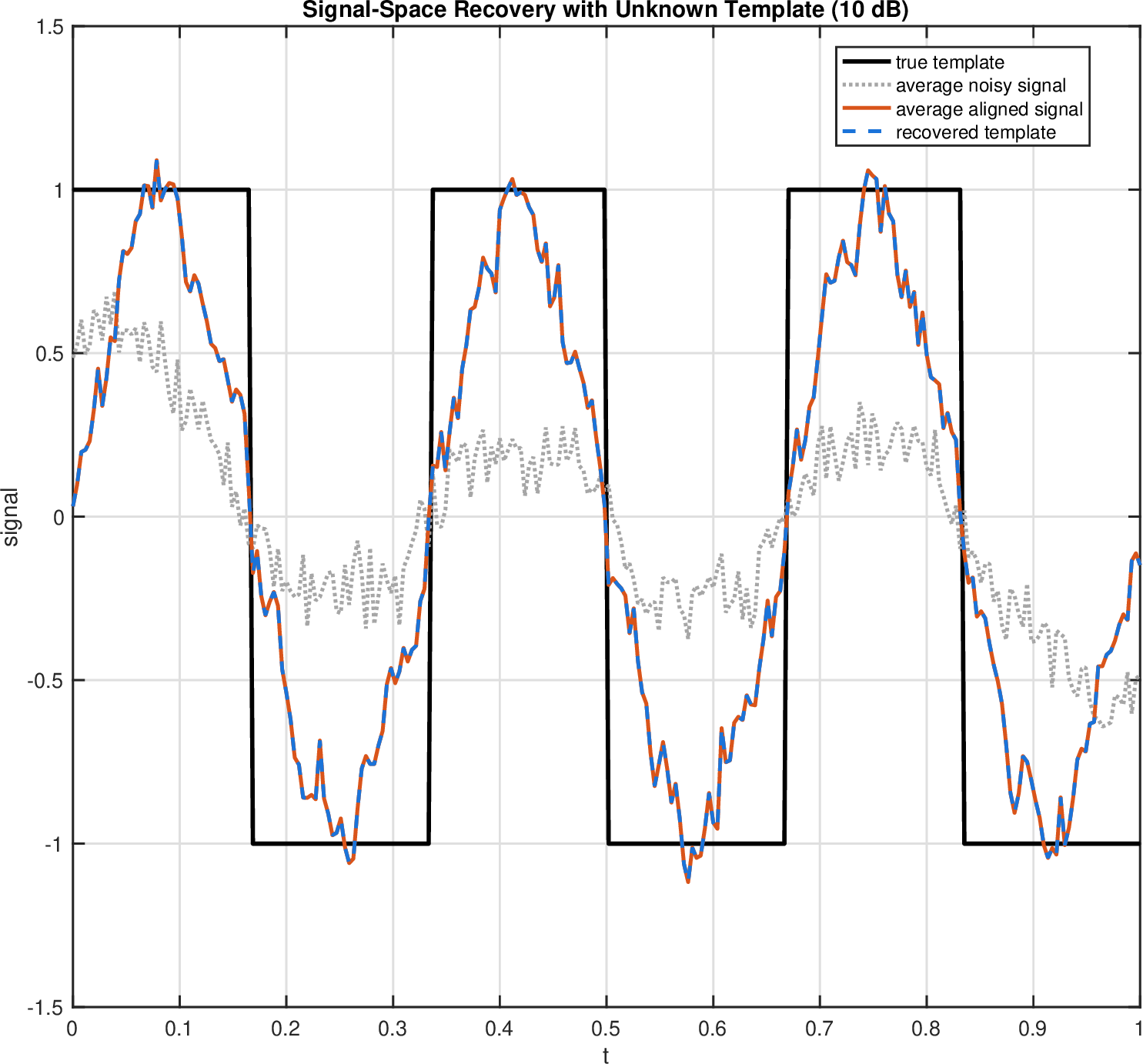}
    \caption{Signal-space interpretation of unknown-template SCDT recovery for translated Gabor, sawtooth, and square signals. The black curve is the true template, the gray curve is the direct average of the noisy shifted signals, and the orange curve is the aligned average, which in this algorithm is also the recovered template estimate.}
    \label{fig:scdt_shift_unknown_gabor_inverse}
\end{figure}

\begin{figure}[h!]
    \centering
    \includegraphics[width=0.6\textwidth]{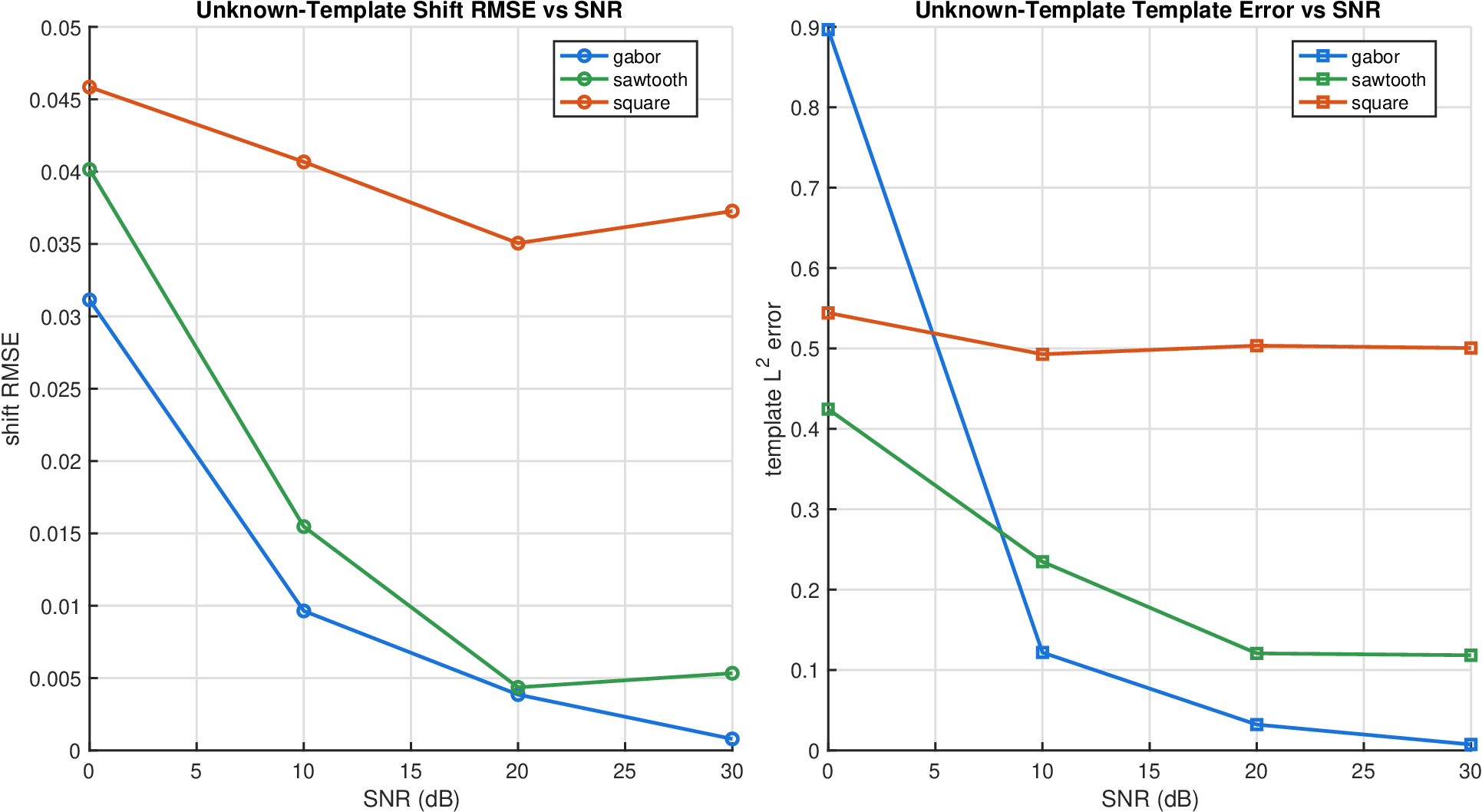}
    \caption{Summary metrics for unknown-template SCDT recovery across all three signal classes. Left: shift RMSE versus SNR. Right: template \(L^2\) error versus SNR.}
    \label{fig:scdt_shift_unknown_summary}
\end{figure}

\bibliographystyle{plain}
\bibliography{refs}

\end{document}